%% file: 0_main.tex
\documentclass[journal]{IEEEtran}
\usepackage{amsthm}
\usepackage{amsmath,amssymb,amsfonts}

\usepackage{graphicx}
\usepackage{textcomp}
\usepackage{xcolor}
\usepackage{booktabs} % For formal tables
\usepackage{array}
\usepackage{multirow}
\usepackage{url}
\usepackage{caption}
\usepackage[caption=false,font=footnotesize]{subfig}
\usepackage{cleveref}
\usepackage{MnSymbol} %package for everywhere/anywhere symbol
\usepackage[normalem]{ulem}
\usepackage[utf8]{inputenc}
\usepackage[T1]{fontenc}
\usepackage[english]{babel}
\usepackage{comment}

\usepackage{multicol} %for double column algorithm
\usepackage{lipsum} %for double column algorithm
\allowdisplaybreaks
\usepackage{caption}%??

\usepackage{mathtools}
\DeclarePairedDelimiter{\ceil}{\lceil}{\rceil}
\DeclareUnicodeCharacter{0301}{\'{e}}
\usepackage[utf8]{inputenc}

\def\BibTeX{{\rm B\kern-.05em{\sc i\kern-.025em b}\kern-.08em
    T\kern-.1667em\lower.7ex\hbox{E}\kern-.125emX}}

\input{macro_lib.tex}
\graphicspath{{Figure/}}

\usepackage[ruled]{algorithm2e} % For algorithms
\SetAlFnt{\small}
\SetAlCapFnt{\small}
\SetAlCapNameFnt{\small}
\SetAlCapHSkip{0pt}
\IncMargin{-\parindent}
\SetKwComment{Comment}{$\triangleright$\ }{}

\SetCommentSty{mycommfont}

\newcommand\newcontent[1]{\color{black}{#1}\color{black}}

\newcommand\revision[1]{\color{black}{#1}\color{black}}

\ifCLASSINFOpdf
\else
\fi
\begin{document}
\title{A Novel Spatial-Temporal Specification-Based Monitoring System for Smart Cities}

\author{Meiyi~Ma,
        Ezio~Bartocci,
        Eli~Lifland,
        John~Stankovic,~\IEEEmembership{Life~Fellow,~IEEE},
        and~Lu~Feng,~\IEEEmembership{Member,~IEEE}
        % <-this % stops a space
\thanks{Corresponding Author: Meiyi Ma, E-mail: meiyi@virginia.edu.}
\thanks{Meiyi Ma, Eli~Lifland, John Stankovic and Lu Feng are with the Department
of Computer Science, University of Virginia, 
VA, 22904. }% <-this % stops a space
\thanks{Ezio Bartocci is with the Faculty of Informatics, TU Wien, Austria.}
% , E-mail: ezio.bartocci@tuwien.ac.at}% <-this % stops a space
% \thanks{Manuscript received ...; revised ...}
% \thanks{Copyright (c) 2021 IEEE. Personal use of this material is permitted. However, permission to use this material for any other purposes must be obtained from the IEEE by sending a request to pubs-permissions@ieee.org.}
}

\maketitle

\begin{abstract}
With the development of the Internet of Things, millions of sensors are being deployed in cities to collect real-time data. This leads to a need for checking city states against city requirements at runtime.   
\newcontent{In this paper, we develop a novel spatial-temporal specification-based monitoring system for smart cities. }
We first describe a study of over 1,000 smart city requirements, some of which cannot be specified using existing logic such as Signal Temporal Logic (STL) and its variants. 
To tackle this limitation, 
we develop SaSTL---a novel Spatial Aggregation Signal Temporal Logic---for the efficient runtime monitoring of safety and performance requirements in smart cities.
We develop two new logical operators in SaSTL to augment STL for expressing spatial aggregation and spatial counting characteristics that are commonly found in real city requirements. We define Boolean and \newcontent{quantitative semantics}~for SaSTL in support of the analysis of city performance across different periods and locations.   
We also develop efficient monitoring algorithms that can check a SaSTL requirement in parallel over multiple data streams (e.g., generated by multiple sensors distributed spatially in a city). \newcontent{Additionally, we build a SaSTL-based monitoring tool to support decision making of different stakeholders to specify and runtime monitor their requirements in smart cities. }
We evaluate our SaSTL monitor by applying it to \newcontent{three case studies}~with large-scale real city sensing data (e.g., up to 10,000 sensors in one study). The results show that SaSTL has a much higher coverage expressiveness than other spatial-temporal logics, and with a significant reduction of computation time for monitoring requirements. 
We also demonstrate that the SaSTL monitor improves the safety and performance of smart cities via simulated experiments. 

\end{abstract}

\begin{IEEEkeywords}
Signal Temporal Logic, Runtime Verification, Smart Cities.
\end{IEEEkeywords}

% \fi
\IEEEpeerreviewmaketitle

\input{1_introduction.tex}

\input{2_framework.tex}

\input{3_requirement.tex}

\input{4_semantics.tex}
\input{5_parallel.tex}
\input{6_tool}
\input{7_evaluation.tex}

\input{8_relatedwork.tex}
\input{9_summary.tex}

\section*{Acknowledgement}
This work was supported in part by National Science Foundation grants CCF-1942836, CNS-1952096 and by the Austrian FFG-funded IoT4CPS project at TU Wien.

\ifCLASSOPTIONcaptionsoff
  \newpage
\fi



\bibliographystyle{IEEEtran}
\bibliography{IEEEabrv,monibib.bib}



% \end{IEEEbiography}

\begin{IEEEbiographynophoto}{Meiyi Ma} is a Ph.D. candidate of Computer Science at the University of Virginia. Her research interests are at the intersection of cyber-physical systems, deep learning and formal methods. 
\end{IEEEbiographynophoto}
\vspace{-1.5cm}

% {\includegraphics[width=1in,height=1.25in,clip,keepaspectratio]{bio/ezio}}

\begin{IEEEbiographynophoto}{Ezio Bartocci} is a full professor at the Faculty of Computer Science of TU Wien, where he leads the Trustworthy Cyber-Physical Systems (TrustCPS) Group. The primary focus of his research is to develop formal methods, computational tools and techniques that support the modeling and the automated analysis of complex computational systems, including software systems, cyber-physical systems and biological systems.
\end{IEEEbiographynophoto}
\vspace{-1.5cm}

\begin{IEEEbiographynophoto}{Eli Lifland} is a software engineer at Ought. He received his Bachelor's degree in Computer Science and Economics at the University of Virginia. 
\end{IEEEbiographynophoto}
\vspace{-1.5cm}

\begin{IEEEbiographynophoto}
{John A. Stankovic} is the BP America Professor in the Computer Science Department
at the University of Virginia and Director of the Link Lab. He is a Fellow of both the IEEE and
the ACM. He has been awarded an Honorary Doctorate from the University of York for his work
on real-time systems. His research interests are in smart and connected health, cyber physical
systems, and the Internet of Things. Prof. Stankovic received his PhD from Brown University.
\end{IEEEbiographynophoto}
\vspace{-1.5cm}

\begin{IEEEbiographynophoto}{Lu Feng} is an Assistant Professor of Computer Science at the University of Virginia. Her research interests are in cyber-physical systems and formal methods. Dr. Feng received her PhD in Computer Science from the University of Oxford in 2014. She is a member of ACM and IEEE.
\end{IEEEbiographynophoto}
% \vspace{-1cm}

%\newpage



%\vfill


\clearpage
\newpage
\input{appendix_iotJ}

\end{document}

%% file: macro_lib.tex
\newtheorem{theorem}{Theorem} 

\newtheorem{definition}{Definition}

\newcommand{\always}{\square}
\newcommand{\eventually}{\lozenge}
\newcommand{\until}{\mathcal{U}_I}

\newcommand{\ew}{\boxbox_{\ra}}
\newcommand{\sw}{\diamonddiamond_{\ra}}
\newcommand{\ag}{\mathcal{A}_{\ra}}
\newcommand{\ct}{\mathcal{C}_{\ra}}

\newcommand{\op}{\mathrm{op}}

\newcommand{\avg}{\mathrm{avg}}

\newcommand{\nb}{L_{\ra}}
\newcommand{\nbx}{\alpha_{\ra}^x(\omega, t, l)}

\newcommand{\sat}{\models}
\newcommand{\eqdef}{\Leftrightarrow}

\newcommand{\everywhere}{\boxbox}
\newcommand{\somewhere}{\diamonddiamond}
\newcommand{\agr}{\mathcal{A}}

\newcommand{\ra}{\mathcal{D}}

\newcommand{\traceset}{\omega}

\newcommand{\tablefontsize}{\scriptsize}

\definecolor{brilliantlavender}{rgb}{0.96, 0.73, 1.0}
\definecolor{blond}{rgb}{0.98, 0.94, 0.75}
\definecolor{celadon}{rgb}{0.67, 0.88, 0.69}
\definecolor{columbiablue}{rgb}{0.61, 0.87, 1.0}
\definecolor{lavenderblush}{rgb}{1.0, 0.94, 0.96}
\definecolor{electriclavender}{rgb}{0.96, 0.73, 1.0}

\newcommand{\ent}[1]{\colorbox{lavenderblush}{#1}}
\newcommand{\temporal}[1]{\colorbox{celadon}{#1}}
\newcommand{\aggregation}[1]{\colorbox{blond}{#1}}
\newcommand{\condition}[1]{\colorbox{pink}{#1}}
\newcommand{\comparison}[1]{\colorbox{columbiablue}{#1}}
\newcommand{\spatial}[1]{\colorbox{electriclavender}{#1}}

\newcommand{\smallfunction}{\delta}

\newcolumntype{L}[1]{>{\raggedright\let\newline\\\arraybackslash\hspace{0pt}}m{#1}}
\newcolumntype{C}[1]{>{\centering\let\newline\\\arraybackslash\hspace{0pt}}m{#1}}
\newcolumntype{R}[1]{>{\raggedleft\let\newline\\\arraybackslash\hspace{0pt}}m{#1}}

\newcommand{\sectref}[1]{Section~\ref{#1}}
\newcommand{\figref}[1]{Figure~\ref{#1}}
\newcommand{\tabref}[1]{Table~\ref{#1}}
\newcommand{\expref}[1]{Example~\ref{#1}}

\newcommand{\algref}[1]{Algorithm~\ref{#1}}

\newtheorem{example}{Example}
\newtheorem{lemma}{Lemma}

\newcommand{\Tset}{\mathbb{T}}

%% file: 1_introduction.tex
% Jack's comments:
% Use word "create", "build", "develop", instead of "present". % Emphasis on "1st work carefully analyze city requirements" and "found"... 
% highlight "challenges" and "contribution"
% following outline as "motivation - challenge - state of art solution and remaining issue - contribution" 

% \vspace{-1.2em}
\section{Introduction}
\label{sect:intro}
Smart cities are emerging around the world. Examples include Chicago's Array of Things project~\cite{arrayofthings}, IBM's Rio de Janeiro Operations Center~\cite{rio2012center} and Cisco's Smart+Connected Operations Center~\cite{cisco-center}, just to name a few. Smart cities utilize a vast amount of data and smart services to enhance the safety, efficiency, and performance of city operations~\cite{ma2019data}. 
%In order to provide real-time services or safety protection, all services and city controllers have to monitor city's safety and performance and thus take actions accordingly at runtime.  For example, governments can use real-time traffic flow data to reduce traffic congestion and manage public transportation;city managers can optimize the energy usage via smart buildings and infrastructures based on the real-time energy demand; law enforcement units can benefit from big data generated from various sensors in the city to keep the citizens safe.
% \cite{fan2018online}. 
% City controllers, such as CityGuard \cite{ma2017cityguard} and CityResolver~\cite{ma2018cityresolver} predict and monitor the changing of city states caused by the services to detect and prevent potential safety violations caused by independent-developed services. The safety violations bring serious consequences to the city environment and citizen's safety. 
% 2.    What is the problem? & Why is it interesting and important? (motivation)
There is a need for monitoring city states in real-time to ensure safety and performance requirements~\cite{ma2017cityguard}.
If a requirement violation is detected by the monitor, the city operators and smart service providers can take actions to change the states, such as improving traffic performance, rejecting unsafe actions, sending alarms to police, etc.
The key \textbf{challenges} of developing such a monitor include how to use an expressive, machine-understandable language to specify smart city requirements, 
and how to efficiently monitor requirements that may involve multiple sensor data streams (e.g., some requirements are concerned with thousands of sensors in a smart city).

Previous works~\cite{zhang2018detecting, sheng2019case, ma2017runtime, haghighi2015spatel} have proposed solutions to monitor smart cities using formal specification languages and their monitoring machinery.
One of the latest works, CityResolver \cite{ma2018cityresolver} uses Signal Temporal Logic (STL)~\cite{Maler2004} to support the specification-based monitoring of safety and performance requirements of smart cities. 
%It considers a scenario where smart services send action requests to the city center, where a simulator tries to predict how the requested actions will change the current city state over a finite horizon of time. In such scenario, STL is used to specify the city requirements that are monitored over the predicted traces.  If there exists a requirement violation, a conflict is detected. CityResolver provides several possible resolution options and predicts the outcome of all these options. 
However, STL is not expressive enough to specify smart city requirements concerning \emph{spatial} information such as \emph{``the average noise level within 1 km of all elementary schools should always be less than 50 dB''}. 
There are some existing spatial extensions of STL (e.g.,  SSTL~\cite{NenziBCLM15}, SpaTeL~\cite{haghighi2015spatel} and 
STREL~\cite{BartocciBLN17,BartocciBLNS20}, \newcontent{see~\cite{NenziBBLV20} for a recent tutorial}), which can express requirements 
such as \emph{``there should be no traffic congestion on 
all the roads in the northeast direction''}. 
%(we refer to~\sectref{sect:related} for a more detailed discussion on the related work)
But they are not expressive enough to specify requirements like \emph{``there should be no traffic congestion on all the roads on average''}, or \emph{``on 90\% of the roads''}, which require the aggregation and counting of signals in the spatial domain. To tackle these challenges and limitations, we develop a novel Spatial Aggregation Signal Temporal Logic (SaSTL), which extends STL with two new logical operators for expressing spatial aggregation and spatial counting characteristics which we demonstrate are commonly found in real city requirements. 
More specifically, this paper has the following major \textbf{contributions}:

\begin{figure*}[t!]
    \centering
    \includegraphics[width = 0.95\textwidth]{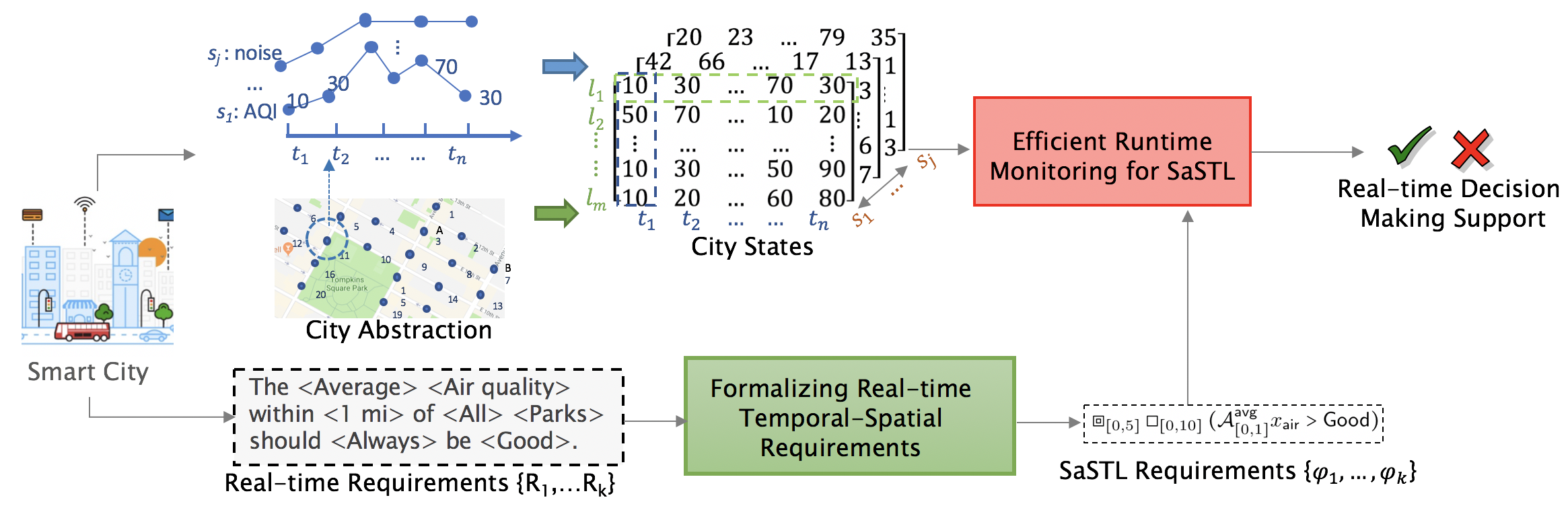}
    \caption{{A framework for runtime monitoring of real-time city requirements}}
    \label{fig:overview}
    \vspace{-1.3em}
\end{figure*}

(1) To the best of our knowledge, this is the first work studying and annotating over 1,000 real smart city requirements across different service domains to identify the gap of expressing smart city requirements with existing formal specification languages. 
As a result, we found that aggregation and counting signals in the spatial domain (e.g., for representing sensor signals distributed spatially in a smart city) are extremely important for specifying and monitoring city requirements. 

(2) Drawing on the insights from our requirements study, we develop a new specification language SaSTL, which extends STL with a \emph{spatial aggregation} operator and a \emph{spatial counting} operator.
SaSTL can be used to specify Point of Interests (PoIs), the physical distance, spatial relations of the PoIs and sensors, aggregation of signals over locations, degree/percentage of satisfaction and the temporal elements in a very flexible spatial-temporal scale. \newcontent{We define Boolean and quantitative semantics with theoretical proofs. }

(3) We compare SaSTL with some existing specification languages and show that SaSTL has a much higher coverage expressiveness (95\%) than STL (18.4\%), SSTL (43.1\%) or STREL (43.1\%) over 1,000 real city requirements.

(4) We develop novel and efficient monitoring algorithms for SaSTL. In particular, we present two new methods to speed up the monitoring performance: (i) dynamically prioritizing the monitoring based on cost functions assigned to nodes of the syntax tree, and (ii) parallelizing the monitoring of spatial operators among multiple locations and/or sensors. 
%We show that both methods improve the time complexity of SaSTL monitoring algorithms. 

(5) We evaluate the SaSTL monitor by applying it to  monitoring real city data collected from Chicago and \newcontent{Aarhus. The results show that SaSTL monitor has the potential to help identify safety violations and support the city managers and citizens to make decisions.} We also evaluate the SaSTL monitor on a third case study of conflict detection and resolution among smart services in simulated New York City with large-scale real sensing data (e.g., up to 10,000 sensors used in one requirement). 
Results of our simulated experiments show that SaSTL monitor can help improve the city’s performance (e.g.,  21.1\%  on  the  environment and 16.6\% on public safety), with a significant reduction of  computation  time  compared  with  previous approaches.

\newcontent{(6) We develop a SaSTL monitoring tool that can support decision making of different stakeholders in smart cities. 
The tool allows users (e.g., city decision maker, citizens) without any formal method background to specify city requirements and monitor city performance easily.  }

This paper is an extended version of \cite{ma2020sastl}. 
\revision{We extend with the following new contributions.
First, we add new quantitative semantics and monitoring algorithms, with new proofs of soundness and correctness in Section III. Compared to the conference version (Boolean semantics), the new quantitative semantics presents the monitoring results with real values, and better supports decision-makers to compare the satisfaction/violation degrees between different options.  
Second, we develop new monitoring algorithms for the proposed quantitative semantics and improve the monitoring algorithms for the new spatial operators in Section IV. 
Third, we develop a monitoring tool to support monitoring and decision making using SaSTL in smart cities in Section VI. The tool also provides a way for non-expert users to input requirements in the English language. Then the tool translates the requirements to SaSTL formal specification automatically for monitoring. 
Fourth, we extend the evaluation with a new city scenario using real data from Aarhus, Denmark in Section VII. The results show that the SaSTL monitor has the potential to help identify safety violations and support city managers and citizens to make decisions. 
Last, we elaborate with more discussions on how to apply the SaSTL monitor in smart cities and extend the related work.
}

%% file: 2_framework.tex
\section{Approach Overview}
\label{sect:framework}

%how it works in the city
\figref{fig:overview} shows an overview of our SaSTL runtime monitoring framework for smart cities. 
We envision that such a framework would operate in a smart city's central control center (e.g., IBM's Rio de Janeiro Operations Center~\cite{rio2012center} or Cisco's Smart+Connected Operations Center~\cite{cisco-center}) where sensor data about city states across various locations are available in real time. 
The framework would monitor city states and check them against a set of smart city requirements at runtime. The monitoring results would be presented to city managers to support decision making. 
The framework makes abstractions of city states in the following way. 
The framework formalizes a set of smart city requirements (See~\sectref{sect:moti}) to some machine checkable SaSTL formulas (See~\sectref{sect:spec}). 
Different data streams (e.g. CO emission, noise level) over temporal and spatial domains can be viewed as a 3-dimensional matrix. For any signal $s_j$ in signal domain $S$, each row is a time-series data at one location and each column is a set of data streams from all locations at one time. Next, the efficient real-time monitoring for SaSTL
verifies the states with the requirements and outputs the Boolean satisfaction to the decision makers, who would take actions to resolve the violation. To support decision making in real time, we improve the efficiency of the monitoring algorithm in \sectref{sect:alg}. \newcontent{We implement SaSTL runtime monitoring tool following this framework for city experts without any formal methods background (see \sectref{sect:tool})}.    We describe more details of the framework in the following sections.

%% file: 3_requirement.tex
\section{Analysis of Real City Requirements}
\label{sect:moti}

\begin{table*}[!h]
	\caption{Examples of city requirements from different domains (The key elements of a requirement are highlighted as, \temporal{temporal}, \spatial{spatial}, \aggregation{aggregation}, \ent{entity}, \condition{condition}, \comparison{comparison}.  )}
	\centering
	\tablefontsize
	\label{tab:reqexamples}
	\begin{tabular}{|L{1.5cm}|L{15cm}|}
		\hline
		\textbf{Domain} & \textbf{Example}\\\hline 
		
		\multirow{3}{*}{\textbf{Transportation}} & \comparison{Limits} \ent{vehicle idling} to \temporal{one minute} adjacent to \aggregation{any} \spatial{school, pre-K to 12th grade}, public or private, in the \spatial{City of New York}~\cite{r1}. \\\cline{2-2}
% 		\cite{r1}
		
		& The engine, power and exhaust mechanism of each motor vehicle shall be equipped, adjusted and operated to \comparison{prevent} the escape of a trail of \ent{visible fumes or smoke} for \comparison{more than} \temporal{ten (10) consecutive seconds}~\cite{r10}.\\\cline{2-2}

		& \comparison{Prohibit} \ent{sight-seeing buses} from using \aggregation{all} \spatial{bus lanes} between the hours of \temporal{7:00 a.m. and 10:00 a.m.} on \temporal{weekdays}~\cite{r2}. \\\hline

		\multirow{2}{*}{\textbf{Energy}} & Operate the \ent{system} to \comparison{maintain} \spatial{zone} temperatures \comparison{down} to 55°F or \comparison{up to} 85°F~\cite{r3}.\\\cline{2-2}
% 		~\cite{r2} )~\cite{r3}
		
		& The \aggregation{total} \ent{leakage} shall be \comparison{less than or equal} to 4 cubic feet \aggregation{per} \temporal{minute} \aggregation{per} \spatial{100 square feet} of \spatial{conditioned floor area}~\cite{r4}. \\\hline
% 		~\cite{r4}
		
		\multirow{2}{*}{\textbf{Environment}}  
% 	&	It shall be \comparison{unlawful} for \aggregation{any} \ent{person}, \temporal{between the hours of 8:00 p.m.} of \aggregation{any} day and \temporal{7:00 a.m.} of the following day to erect, construct, demolish, excavate for, alter or repair \aggregation{any} \spatial{building or structure} \condition{if} the noise level created thereby is \comparison{in excess of} the ambient noise level by 5 dB at the  \spatial{nearest property plane},  \condition{unless} a special permit therefor has been applied for and granted by the Director of Public Works or the Director of Building Inspection~\cite{r5}.\\\cline{2-2}
		
		& LA Sec 111.03 \comparison{minimum} \ent{ambient noise level} table: \spatial{ZONE M2 and M3} -- \temporal{DAY}: 65 dB(A) \temporal{NIGHT}: 65 dB(A)~\cite{r11}. \\\cline{2-2}
		
		 & The \aggregation{total amount} of \ent{HCHO emission} should be \comparison{less than} 0.1mg \aggregation{per} m$^3$ \temporal{within an hour}, and the \aggregation{total amount} of PM10 emission should be \comparison{less than} 0.15 mg \aggregation{per} m$^3$ \temporal{within 24 hours}~\cite{r6}. \\\hline
% 		~\cite{r6}
	\multirow{2}{*}{\textbf{Emergency}}  & NYC Authorized \ent{emergency vehicles} may \condition{disregard} 4 primary rules regarding traffic~\cite{r7}.  \\\cline{2-2}
% 		~\cite{r7}
	 & \comparison{At least} one \ent{ambulance} should be equipped \aggregation{per} 30,000 population (counted \spatial{by area}) to obtain the shortest radius and fastest response time~\cite{r8}. \\\hline
% 		 ~\cite{r8}
	\textbf{Public Safety} &  \ent{Security staff} shall visit \comparison{at least once} \temporal{per week} in \spatial{public schools}~\cite{r9}. \\\hline
% 	 ~\cite{r9}
	\end{tabular}
	\vspace{-1.4em}
\end{table*}

To better understand real city requirements, we conduct a requirement study. 
\revision{We collect and statistically analyze 1000 quantitatively specified city requirements 
(e.g., standards, regulations, city codes, and laws) across different application domains, including energy, environment, transportation, emergency, and public safety from over 70 cities (e.g. New York City, San Francisco, Chicago, Washington D.C., Beijing, etc.) around the world. } Some examples of these city requirements are highlighted in \tabref{tab:reqexamples}. We identify key required features to have in a specification language and its associated use in a city runtime monitor. The summarized statistical results of the study and key elements we identified (i.e., temporal, spatial, aggregation, entity, comparison, and condition) are shown in \tabref{tab:elements}.

\begin{table}[t]
	\caption{Key elements of city requirements and statistical results from 1000 real city requirements}
	\centering
	\label{tab:elements}
	\tablefontsize
\begin{tabular}{|L{1.2cm}|l|c|L{3.2cm}|}
\hline
\textbf{Element }                                                                & \textbf{Form}             & \textbf{Number} & \textbf{Example}                                           \\ \hline
\multirow{4}{*}{\begin{tabular}[c]{@{}l@{}}\textbf{Temporal}\\ \end{tabular}}    & Dynamic Deadline &  77   & limit ... to one minute                           \\ \cline{2-4} 
                                                                            & Static Deadline  &   98  & at least once a week                              \\ \cline{2-4} 
                                                                            & Interval         &  168   & from 8am to 10am; within 24 hours;                \\ \cline{2-4} 
                                               
                                                                            & Default          &  657   & The noise (always) should not exceed 50dB.        \\ \hline
\multirow{3}{*}{\begin{tabular}[c]{@{}l@{}}\textbf{Spatial}\\ \end{tabular}}     & \textbf{PoIs/Tags}             &  \textbf{801}   & school area; all parks;                      \\ \cline{2-4} 
                                                                            & \textbf{Distance}         &  \textbf{650}   & Nearby                                            \\ \cline{2-4} 
                                                                            & Default          &  154   & (everywhere) ; (all) locations                    \\ \hline
\multirow{3}{*}{\begin{tabular}[c]{@{}l@{}}\textbf{Aggregation}\\ \end{tabular}} & \textbf{Count, Sum}       &  \textbf{256}   & in total; x out of N locations; \%;                    \\ \cline{2-4} 
                                                                            & \textbf{Average}          &  \textbf{196}   & per m$^2$;                                        \\ \cline{2-4} 
                                                                            & \textbf{Max, Min}         &  \textbf{67}   & highest/lowest value                              \\ \hline                               
\textbf{Entity} & Subject & 1000 & air quality; Buses; \\\hline
\multirow{3}{*}{\begin{tabular}[c]{@{}l@{}}\textbf{Comparison}\\ \end{tabular}}     & Value comparison &  836   & More than, less than                              \\ \cline{2-4} 
                                                                            & Boolean          &  388   & Street is blocked; should                                \\ \cline{2-4} 
                                                                            & Not              &   456  & It is unlawful/prohibited...                       \\ \hline
                                         
\multirow{2}{*}{\begin{tabular}[c]{@{}l@{}}\textbf{Condition}\\ \end{tabular}}  
                                                                            & Until            &    24 & keep... until the street is not blocked.          \\ \cline{2-4} 
                                                                            % & And              &     & ...should... and ...should...                     \\ \cline{2-4} 
                                                                            & If/Except           &  44   & If rainy, the speed limit...  \\ \hline
\end{tabular}
\vspace{-2em}
\end{table}

% point of interest
% range
% count / aggregate 

% Data streams from multiple locations are monitored spontaneously, which raise the demand for monitor the hyper properties across spatial 

% Third, the syntax of city requirements from different domains is composed of six key elements, i.e. temporal, spatial, aggregation, entity, comparison, and condition. 
% The semantics of the requirements are less domain specific. These elements and patterns are equally important in all five application domains, which indicates that the templates to specify requirements are more reasonable to be categorized by the semantics, instead of by the application domain.  
% Six key elements forming a city requirement, . 

\textbf{Temporal}: Most of the requirements include a variety of temporal constraints, e.g. a static deadline, a dynamic deadline, or time intervals. In many cases (65.7\%), the temporal information is not explicitly written in the requirement, which usually means it should be ``always'' satisfied. 
In addition, city requirements are highly real-time driven. In over 80\% requirements, cities are required to detect requirement violations at runtime. It indicates a high demand for runtime monitoring. 

% We also found that the city requirements are highly real-time driven.
% To enhance safety and performance, cities are required to detect and take actions to resolve requirement violations, or even predict potential violations and take precautions to prevent such violations in real time, which requires that the violations should be detected before a certain deadline. 
% Usually, it is a deadline like ``within one minute'', or time interval like ``between the hours of 7:00 a.m. and 10 a.m.'',  ``after 9:00 p.m.'' or ``weekdays''. 

\textbf{Spatial}: A requirement usually specifies its spatial range explicitly using the Points of Interest (PoIs) (80.1\%), such as  ``park'', ``xx school'', along with a distance range (65\%). One requirement usually points to a set of places (e.g. all the schools). Therefore, it is very important for a formal language to be able to specify the spatial elements across many locations within the formula, rather than one formula for each location. 

We also found that the city requirements specify a very large spatial scale. Different from the requirements of many other CPS, requirements from smart cities are highly spatial-specific and usually involve a very large number of locations/sensors. 
For example, the first requirement in \tabref{tab:reqexamples} specifies a vehicle idling time ``adjacent to any school, pre-K to 12th grade in the City of New York''. There are about 2000 pre-K to 12th schools, even counting 20 street segments nearby each school, there are 40,000 data streams to be monitored synchronously. An efficient monitoring is highly demanded.

\textbf{Aggregation}: In 51.9\% cases, requirements are specified on the aggregated signal over an area, such as, ``the total amount'', ``average...per 100 square feet'', ``up to four vending vehicles in any given city block'', ``at least 20\% of travelers from all entrances should ...'', etc. Therefore, aggregation is a key feature for the specification language.  
% The same set of data streams can be checked on different requirements with different ways of aggregation depending on the context. It is not practical to aggregate the data beforehand. Therefore, aggregation is a key feature for the specification language.  

% \textbf{Entity}: An entity requirement specifies the variable of interest, such as, ``noise level'', ``energy consumption'', etc.

% \textbf{Comparison}: Comparison requirements usually specify the threshold of the variable, e.g., ``up to 85°F''. In other cases, it also defines true or false, e.g.,  ``the street is blocked''.

% \textbf{Condition}: In a broad definition, conditions specify the condition or special cases of the requirement, such as, ``if/unless'', ``until'', and ``except''. 

%% file: 4_semantics.tex
% \section{Spatial Aggregation Signal Temporal Logic}
\section{Formalizing Temporal-Spatial Requirements}
\label{sect:spec}

SaSTL extends STL with 
two spatial operators: a \emph{spatial aggregation} operator
and a \emph{neighborhood counting} operator.
Spatial aggregation enables combining (according to a chosen 
operation) measurements of the same type (e.g., environmental 
temperature), but taken from different locations.
The use of this operator can be suitable in requirements 
where it is necessary to evaluate the average, best or worst 
value of a signal measurement in an area close to the desired location.
The neighborhood counting operator allows measuring the number/percentage of neighbors of a location that satisfy a certain requirement. 
% In this section, we formally define the syntax and semantics.

% $\everywhere_{[0,5]}\always_{[0,10]}(\mathcal{A}_{[0,1]}^\mathsf{avg}x_\mathsf{air}>\mathsf{Good})$

\subsection{SaSTL Syntax}
We define a multi-dimensional \emph{spatial-temporal signal} as $\omega: \mathbb{T} \times L \to \{{\mathbb{R}\cup\{\bot\}\}} ^n$,
where $\mathbb{T}=\mathbb{R}_{\ge 0}$, represents the continuous time and $L$ is the set of locations. We define $X= \{x_1, \cdots, x_n \}$ as the set of 
variables for each location. 
Each variable can assume a real value $ v \in \mathbb{R}$ or
is undefined for a particular 
location ($x_i = \bot$).
%which includes real-valued, Boolean and \emph{null} signals. 
We denote by $\pi_{x_{i}}(\omega)$ as the projection of $\omega$ on its component variable $x_i \in X$.  We define $P = \{p_1, \cdots, p_m \}$ a set of propositions (e.g. $\{\mathsf{School, Street, Hospital}, \cdots\}$ ) and 
$\mathcal{L}$ a labeling function $\mathcal{L}: L \rightarrow 2^{P}$ that assigns for each location 
the set of the propositions that 
are true in that location. 

% weighted graph
A \emph{weighted undirected graph} is a tuple $G=(L, E, \eta)$ where 
$L$ is a finite non-empty set of nodes representing locations,
$E \subseteq L \times L$ is the set of edges connecting nodes,
and $\eta: E \to \mathbb{R}_{\ge 0}$ is a cost function over edges.
We define the \emph{weighted distance} between two locations $l, l' \in L$ as
$$ d(l,l'):= \min\{\sum_{e\in \sigma}\eta(e) \ |\ \sigma \mbox{ is a path between } l \mbox{ and } l'\}. $$

Then we define the spatial domain $\ra$ as, 

\begin{equation*}
\begin{array}{cl}
     \ra :=&  ([d_1, d_2],\psi) \\
     \psi :=& \top\;|\;p\;|\;\neg\;\psi\;|\;\psi\;\vee\;\psi 
\end{array}
    % \;|\
%  p \;|\
% \ra_1 \land \ra_2 \;|\
%  \ra_1 \lor \ra_2 \\
\end{equation*}
 
\revision{where $[d_1, d_2]$ defines a spatial interval with $d_1 < d_2$ and $d_1,d_2 \in \mathbb{R}$, and $\psi$ specifies the property 
over the set of propositions
that must hold in each location. Intuitively, it draws two circles with radius $r_1=d_1$ and $r_2=d_2$, and the locations $l\sat \psi$ between these two circles are selected. }
% We define a labeling function $\mathcal{L}: L \times P \rightarrow \{0,1\}$, $\mathcal{L}(l,p)=1$ if and only if location $l$ has a label $p$.
In particular, 
$\ra = ( [0,+\infty), \top)$ indicates the whole spatial domain. 
% We denote by $\nb^l:=\{l' \in L \ |\ l' \sat \ra \}$ as 
% the set of locations satisfying $\ra$. 
% More specifically, 
% $L_{[d_1,d_2]}^l:=\{l' \in L \ | \ 0 \le d_1 \le d(l, l') \le d_2\}$ as the set of locations at a distance between $d_1$ and $d_2$ from $l$, and 
We denote $L_{([d_1,d_2], \psi)}^l:=\{l' \in L \ | \ 0 \le d_1 \le d(l, l') \le d_2 \mbox{ and } 
% \mathcal{L}(l',\psi)=1\}$
\mathcal{L}(l') \sat \psi \}$ 
as the set of locations at a distance between $d_1$ and $d_2$ from $l$ for which $\mathcal{L}(l')$ satisfies $\psi$.    
We denote the set of non-null values for signal variable $x$ at time point $t$ location $l$ over locations in $\nb^l$ by 
\begin{equation*}
\scriptsize
 \nbx:=\{\pi_x(\omega)[t, l'] \ | \ l' \in \nb^l \mbox{ and } \pi_x(\omega)[t, l'] \neq \bot\}.   
\end{equation*}
We define a set of operations $\op(\nbx)$ for $\op \in \{\max, \min, \mathrm{sum}, \avg\}$  when $\nbx \neq \emptyset$ 
that computes the maximum, minimum, summation and average of values in the set $\nbx$, respectively. 
% remark about how to construct the graph
To be noted, Graph $G$ and its weights between nodes are constructed flexibly based on the property of the system. For example, we can build a graph with fully connected sensor nodes and their Euclidean distance as the weights when monitoring the air quality in a city; or we can also build a graph that only connects the street nodes when the two streets are contiguous and apply Manhattan distance. 
It does not affect the syntax and semantics of SaSTL. 

%syntax

The syntax of SaSTL is given by 
\begin{equation*} 
\begin{array}{cl}
% \loc :=& \neg \loc \;| \loc_1 \land \loc_2\\
% \mu :=& x \sim c\;|  \ag^{\op} x \sim c \\
\varphi  :=& x \sim c\;| \
 \neg \varphi \;| \
\varphi_1 \land \varphi_2 \;|\
\varphi_1 \until \varphi_2 \;|\
% \select^{\loc} \varphi \;| \\
\ag^{\op} x \sim c \;|\
\ct^{\op} \varphi \sim c \\
% \ra := & [d_1, d_2]\;|\
%  p \;|\
% \ra_1 \land \ra_2 \;|\
%  \ra_1 \lor \ra_2 \\
\end{array}
% \varphi  :=   p\;|  
%   x \sim c\;|  
%  \neg \varphi \;| 
% \varphi_1 \land \varphi_2 \;|
% \varphi_1 \until \varphi_2 \;| 
% % \varphi_1 \past \varphi_2 \ |
% %  \everywhere_{[d_1,d_2]} \varphi \;| 
%  \ag^{\op} x \sim c \;| 
% \ct^{\op} \varphi \sim c
\end{equation*}

where $x \in X$, $\sim \in \{<, \le\}$, $c \in \mathbb{R}$ is a constant, 
$I \subseteq \mathbb{R}_{> 0}$ is a real positive dense time interval, 
\revision{$\until$ is the \emph{bounded until} temporal operators from STL. 
The \emph{always} (denoted $\always$) and \emph{eventually} (denoted $\eventually$) temporal operators can be derived the same way as in STL, where $\eventually \varphi \equiv \mathsf{true} \ \until \varphi$, and $\always \varphi \equiv \neg \eventually \neg \varphi$.}

% We define 
% $\everywhere_{[d_1,d_2]} \varphi = \op\{(\omega, t, l')  \sat \varphi \ | \ l' \in \nb^l\}.$

In SaSTL, we define a set of spatial \emph{aggregation} operators $\ag^{\op} x \sim c$ for 
$\op \in \{\max, \min, \mathrm{sum}, \avg\}$
that evaluate the aggregated product of traces $\op(\nbx)$ over a set of locations $l \in \nb^l$.
We also define a set of new spatial \emph{counting} operators $\ct^{\op} \varphi \sim c$ for 
$\op \in \{\max, \min, \mathrm{sum}, \avg\}$ that counts the satisfaction of traces over a set of locations. 
More precisely, we define 
$\ct^{\op} \varphi = \op(\{g((\omega, t, l')  \sat \varphi) \ | \ l' \in \nb^l\})$, where $g((\omega, t, l)  \sat \varphi)) = 1$ if $(\omega, t, l)  \sat \varphi$, otherwise $g((\omega, t, l)  \sat \varphi)) = 0$. 
From the new \textit{counting} operators, we also derive the \emph{everywhere} operator as $\ew \varphi \equiv \ct^{\mathrm{min}} \varphi > 0$, and \emph{somewhere} operator as $\sw \varphi \equiv \ct^{\mathrm{max}} \varphi > 0$.
% $\sw \varphi \equiv \neg \ew \neg \varphi $.
In addition, $\ct^{\mathrm{sum}} \varphi$ specifies the total number of locations that satisfy $\varphi$ and $\ct^{\mathrm{avg}} \varphi$ specifies the percentage of locations satisfying $\varphi$. 
% These are all very important spatial operations

We now illustrate how to use SaSTL to specify various city requirements, especially for the spatial aggregation and spatial counting, and how important these operators are for the smart city requirements using examples below.

\begin{example} [Spatial Aggregation]
Assume we have a requirement, ``The average noise level in the school area (within 1 km) in New York City should always be less than 50 dB and the worst should be less than 80 dB in the next 3 hours'' is formalized as,
% \begin{equation*}
% \noindent
   $\everywhere_{([0, +\infty), \mathsf{School})}\always_{[0,3]}( (\agr_{([0,1], \top)}^{\avg} x_\mathsf{Noise} < 50) \land (\agr_{([0,1], \top)}^{\mathsf{max}} x_\mathsf{Noise} < 80))$.
% \end{equation*}
$([0, +\infty), \mathsf{School})$ selects all the locations labeled as ``school'' within the whole New York city (${[0, +\infty)}$) (predefined by users). 
$\always_{[0, 3]}$ indicates this requirement is valid for the next three hours. $(\agr_{([0,1], \top)}^{\avg} x_\mathsf{Noise} < 50) \land (\agr_{([0,1], \top)}^{\mathsf{max}} x_\mathsf{Noise} < 80)$ calculates the average and maximal values in 1 km for each ``school'', and compares them with the requirements, i.e. 50 dB and 80 dB. 
\label{ex:sytax}
\vspace{-0.5em}
\end{example}

Without the spatial aggregation operators, STL and its extended languages cannot specify this requirement. First, they are not able to first dynamically find all the locations labeled as ``school''. 
To monitor the same spatial range, users have manually get all traces from schools, and then repeatedly apply this requirement to each located sensor within 1 km of a school and do the same for all schools.  
More importantly, STL and its extended languages could not specify ``average'' or  ``worst'' noise level. Instead, it only monitors each single value, which is prone to noises and outliers and thereby causes inaccurate results.

\begin{example}[Spatial Counting]
A requirement that ``At least 90\% of the streets, the particulate matter (PMx) emission should not exceed \textit{Moderate} in 2 hours'' is formalized as
% \begin{equation*}
  $\mathcal{C}_{([0,+\infty), \mathsf{Street})}^\avg(\always_{[0,2]} ( x_\mathsf{PMx} < \mathsf{Moderate})) > 0.9$.
% \end{equation*}
${\mathcal{C}_{([0,+\infty),\mathsf{Street})}^\avg} \varphi > 0.9$ represents the percentage of satisfaction is larger than 90\%. 
Specifying the percentage of satisfaction is very common and important among city requirements. 
% Only SaSTL can specify this kind of requirement over a spatial domain. 
\label{ex2:syntax}
\end{example}

% $\mathcal{C}_{([0,+\infty), \mathsf{Street})}^\avg(\always_{[0,2]} ( x_\mathsf{PMx} < \mathsf{Moderate})) > 0.9$

% \begin{example}[Spatial-Counting]
% Assuming a requirement that `` If accidents happened, traffic nearby (0.5km) should be moderate on average and safe in worst case. ''
% \begin{equation*}
% \begin{split}
%     & \mathsf{Accident} \rightarrow \boxbox_{ [0, +\infty)} \always_{[0,+\infty)} (\agr_{[0,0.5]}^{\avg} x_\mathsf{traffic} > 
%      \mathsf{Moderate} \\
%     &\land \agr_{[0,0.5]}^{\max} x_\mathsf{traffic} > \mathsf{Safe})
% \end{split}
% \end{equation*}

% \label{ex:spatalcounting2}
% \end{example}

% If using STL (CityResolver), users can only specify a range of the variable should stay in, like ``$x_{\text{traffic}} < \mathsf{Moderate}$ '' or ``$x_{\mathsf{traffic}}  < \mathsf{Safe} $'', which only evaluates one location (or multiple locations independently) with a single range. 
% It is not specific and accurate enough and 
% thus violations detected by this requirement are less meaningful. 

\subsection{SaSTL Semantics}
We define the SaSTL semantics as the \emph{satisfiability relation}
$(\omega, t, l) \sat \varphi$, indicating that the spatio-temporal signal $\omega$ satisfies a formula $\varphi$ at the time point $t$ in location $l$ when $\pi_v(\omega)[t, l]\neq \bot$ and $\nbx \neq \emptyset$.  We define that $(\omega, t, l)  \sat \varphi$ if $\pi_v(\omega)[t, l] = \bot$.
% according to the following definitions. 
% Note that the semantics are only well-defined when $\pi_v(\omega)[t, l]\neq \bot$ and $\nbx \neq \emptyset$.

\vspace{-1em}
% {\small
\begin{alignat*}{2}
    % (\omega, t, l) & \sat p 
        % && \eqdef \pi_p(\omega)[t, l] = \mathsf{true}  \\
        % (\omega, t, l) & \sat \loc
        % && \eqdef \pi_\loc(\omega)[t, l] = \mathsf{true}\\
        % (\omega, t, l) & \sat \neg \loc
        % && \eqdef (\omega, t, l)  \not \sat \loc\\
        % (\omega, t, l) & \sat \loc_1 \land \loc_2
        % && \eqdef (\omega, t, l) \sat \loc_1 \mbox{ and } (\omega, t, l)  \sat \loc_2\\
    (\omega, t, l) & \sat x \sim c 
        && \eqdef \pi_x(\omega)[t, l] \sim c  \\
    (\omega, t, l) & \sat \neg \varphi
        && \eqdef (\omega, t, l) \not \sat \varphi \\
    (\omega, t, l) & \sat \varphi_1 \land \varphi_2
        && \eqdef (\omega, t, l) \sat \varphi_1 \mbox{ and } (\omega, t, l) \sat \varphi_2 \\
    (\omega, t, l) & \sat \varphi_1 \until \varphi_2 
        && \eqdef \exists t' \in (t+I) \cap \mathbb{T}: (\omega, t', l) \sat \varphi_2 \\
& &&         \mbox{ and } \forall t'' \in (t, t'), (\omega, t'', l) \sat \varphi_1 \\
    % (\omega, t, l) & \sat \varphi_1 \past \varphi_2 
    %     && \eqdef \exists t' \in (t-I) \cap \mathbb{T}: (\omega, t', l) \sat \varphi_2 \\
        %   & && \hspace{2em}     \mbox{ and } \forall t'' \in (t', t), (\omega, t'', l) \sat \varphi_1 \\    
%   (\omega, t, l) & \sat \select^\loc \varphi
%         && \eqdef \forall l' \in \nb^l  \mbox{ and } \pi_\loc(\omega)[t, l']: (\omega, t, l')  \sat \varphi   \\
    (\omega, t, l) & \sat  \ag^{\op} x \sim c   
        && \eqdef \op (\nbx  ) \sim c  \\
    (\omega, t, l) & \sat  \ct^{\op} \varphi \sim c   
        && \eqdef     \op(\{g((\omega, t, l')  \sat \varphi) \ | \ l' \in \nb^l\}) \sim c \\
\end{alignat*} 
% }

% \vspace{-1cm}
\newcontent{where, for counting operator $(\omega, t, l) \sat  \ct^{\op} \varphi \sim c$, the valid ranges for $c$ are 
$c\in [0,1)$ when $\mathsf{op = sum/min}$, and 
$c\in [0,N]$ when $\mathsf{op = sum/min}$. Otherwise (e.g., $c<0$), the requirement is trivially satisfied or violated. 
% example 5: calculating semantics
\begin{example}
Following \expref{ex:sytax}, checking the city states with a requirement, 

$\everywhere_{([0, +\infty), \mathsf{Hospital})}\always_{[0,5]}( (\agr_{([0,500],\top)}^{\avg} x_\mathsf{AQI} < 50) \land (\agr_{([0,500],\top)}^{\mathsf{max}} x_\mathsf{AQI} < 80))$, 

to start with, assuming we have the AQI level data from a number of sensors within 500 meters of one of the hospital, the sensor readings in 5 hours as, \{[51, ..., 11], [80, ..., 30],..., [40, ..., 30]\}, $\varphi_t = (\agr_{([0,500],\top)]}^{\avg} x_\mathsf{AQI} < 50) \land (\agr_{([0,500],\top)}^{\mathsf{max}} x_\mathsf{AQI} < 80)$, then, we check $\varphi_t$ for this hospital at each time,

at $t = 1$, $\avg(51, ..., 40) > 50 \land \mathsf{max}(51, ..., 40) < 80$,
thus, $\varphi_{t1} = False$,

at $t = 2$, $\avg(49, ..., 20) < 50 \land \mathsf{max}(49, ..., 20) > 80$, thus, $\varphi_{t1} = False$,

...

at $t = 5, \avg(11, ..., 30) < 50 \land \mathsf{max}(11, ..., 30) < 80$,
thus, $\varphi_{t1} = True$.

Thus, we have $\always_{[0,5]}\varphi_t = False$. 

Next, the monitor checks all qualified hospitals the same way and reaches the final results, 

$\everywhere_{([0, +\infty),\mathsf{Hospital})}\always_{[0,5]} ((\agr_{([0,500],\top)}^{\avg} x_\mathsf{AQI} < 50) \land (\agr_{([0,500],\top)}^{\mathsf{max}} x_\mathsf{AQI} < 80)) = False$. 

\label{ex:quali}
\hspace{0.5cm}
\end{example}

In a real scenario, the monitor algorithm can also decide to terminate the monitor and return the False result when at $t=1$, because the always operator returns False as long as a one-time violation occurs. Similarly, the everywhere operator will also return False when the first hospital violates the requirement. }

\newcontent{
\begin{definition}[Quantitative Semantics]

Let $x > c$ be a numerical predicate,
we then define the robustness degree (i.e. 
the quantitative satisfaction) function $\rho(\varphi, \omega, t, l)$
for an SaSTL formula over a spatial-temporal signal $\omega$ as,

\begin{equation*}
\small
\begin{array}{ll}
\rho(x \sim c, \omega, t, l)&= \pi_x(\omega)[t, l] - c \\
  \rho(\neg \varphi, \omega, t, l)&= - \rho(\varphi, \omega, t, l)\\
  \rho(\varphi_1 \vee \varphi_2, \omega, t, l)  &=  \max\{\rho(\varphi_1, \omega, t, l), \rho(\varphi_2, \omega, t, l) \}\\
\rho(\varphi_1 \until \varphi_2, \omega, t, l)  & =\sup_{t'\in (t + I) \cap \mathbb{T}} (\min\{\rho(\varphi_2, \omega, t', l), \\ & \inf_{t''\in[t,t']}(\rho(\varphi_1, \omega, t'', l)) \})\\
\\
\rho(\ag^{\op} x \sim c, \omega, t, l)& =\begin{cases}
            \frac{\mathsf{sum} (\nbx) - c}{|\nbx|} ~~~& \mathsf{op= sum} \\
            \op (\nbx) - c ~~& \mathsf{op} \in \{ \mathsf{max, min, avg}\}
        \end{cases} \\ 
\rho(\ct^{\op} \varphi \sim c, \omega, t, l)  &\\
\end{array}
\end{equation*}
\begin{equation*}
% \small
\small
=\begin{cases}
        \max_{l' \in \nb^l}\{\rho(\varphi,\omega,t,l') \} & \mathsf{op = max} \\
        \min_{l' \in \nb^l}\{\rho(\varphi,\omega,t,l') \} & \mathsf{op = min}\\
        \smallfunction(\ceil[\big]{c}, \{\rho(\varphi,\omega,t,l') \ | \ l' \in \nb^l \}) & \mathsf{op = sum}\\
        \smallfunction(\ceil[\big]{c \times |\nb^l |}, \{\rho(\varphi,\omega,t,l') \ | \ l' \in \nb^l \}) & \mathsf{op = avg}\\
        % +\infty & c \le 0\\
        % , c\in [0,1)
        \end{cases} 
\end{equation*}

\end{definition}

where we define $\smallfunction(k,S)$ as a function that returns the $k$th smallest number of set $S$, $|S|>0$, and $0 \leq k \leq |S|$. 
% Formally, 
% \begin{equation}
%     \smallfunction = \begin{cases}
%     \min_{x \in S} x ~~~~~~~~ & k\leq 1 \\
%       \min_x(\left|\{ x'| x'<x \land x' \in S\}\right| > k) ~~~~~~~~~~ & 1<k \leq |S|\\
%       -\infty ~~~~~~~~~& k>|S|
%     \end{cases}
% \end{equation}
For $\ct^{\op} \varphi \sim c$, when $\mathsf{op} = \mathsf{sum}$, it requires that there are at least $\ceil[\big]{c}$ locations that satisfy $\varphi$, thus, we denote the $\ceil[\big]{c}$th smallest robustness value from $\{\rho(\varphi,\omega,t,l') \ | \ l' \in \nb^l \}$ as the robustness value of this formula. $[c]$ indicates the smallest integer that is larger than or equal to $c$. 
Similarly, when $\mathsf{op} = \mathsf{avg}$, the formula is converted as there are at least $\ceil[\big]{c\times |\nb^l |}$ locations that satisfy $\varphi$, thus, we denote the $\ceil[\big]{c\times |\nb^l |}$th smallest robustness value from $\{\rho(\varphi,\omega,t,l') \ | \ l' \in \nb^l \}$ as the robustness value of this formula.   
Same as the Boolean semantics, the valid ranges for $c$ are 
$c\in [0,1)$ when $\mathsf{op = sum/min}$, and 
$c\in [0,N]$ when $\mathsf{op = sum/min}$. Otherwise (e.g., $c<0$), the requirement is trivially satisfied or violated. 

% \textcolor{blue}{Note that we cannot define the quantitative semantics as $c - \op(\{g((\omega, t, l')  \sat \varphi) \ | \ l' \in \nb^l\})$ because we cannot prove the correctness. "everywhere" and "somewhere" are not systematic with "always" "eventually". Instead, $[c]$th smallest is consistent with the temporal operators.}

\begin{example}
% Following \expref{ex:quali}, we calculate the satisfaction below, 
Assuming we have data (1,2,3), (2,3,4), (4,5,7) from three locations satisfying $\mathcal{D}$, thus, 

\begin{itemize}
\small
    \item $\rho(\ct^\mathsf{max}(\always_{[0,2]}(x>5))>0) = \rho (\ct^\mathsf{max}(\{-4,-3,2\})>0) = 2$
    \item $\rho(\ct^\mathsf{min}(\always_{[0,2]}(x>5))>0) = \rho (\ct^\mathsf{min}(\{-4,-3,2\})>0) = -4$
    \item $\rho(\ct^\mathsf{sum}(\always_{[0,2]}(x>5))>1) = \rho (\ct^\mathsf{sum}(\{-4,-3,2\})>1) = -3$
    \item $\rho(\ct^\mathsf{avg}(\always_{[0,2]}(x>5))>0.2) = \rho (\ct^\mathsf{avg}(\{-4,-3,2\})>0.2) = 2$
\end{itemize}

% $\ct^\mathsf{max}(\always_{[0,2]}(x>5))>0$, assuming we have data (1,2,3), (2,3,4), (4,5,7) from three locations in D, thus,  $\ct^\mathsf{max}(\{-4,-3,2\})= 2$. 

% $\ct^\mathsf{min}(\always_{[0,2]}(x>5))>0$, assuming we have data (1,2,3), (2,3,4), (4,5,7) from three locations in D, thus,  $\ct^\mathsf{min}(\{-4,-3,2\})= -4$. 

% $\ct^\mathsf{sum}(\always_{[0,2]}(x>5))>0$, assuming we have data (1,2,3), (2,3,4), (4,5,7) from three locations in D, thus,  $\ct^\mathsf{sum}(\{-4,-3,2\})= 2$. 

% $\ct^\mathsf{sum}(\always_{[0,2]}(x>5))>1$, assuming we have data (1,2,3), (2,3,4), (4,5,7) from three locations in D, thus,  $\ct^\mathsf{sum}(\{-4,-3,2\})= -3$. 

% $\ct^\mathsf{avg}(\always_{[0,2]}(x>5))> 0.2$, assuming we have data (1,2,3), (2,3,4), (4,5,7) from three locations in D, thus,  $\ct^\mathsf{avg}(\{-4,-3,2\})= 2$. 

% at $t = 1$, $50 - \avg(51, 80, 40) = -7$,
% $\rho({\agr_{[0,1]}^{\avg} x_{\mathsf{air}} > \mathsf{Good}}) = -7$

% at $t = 2$, $ 50 - \avg(49, 20, 20) = 20.7$, 
% $\rho(\agr_{[0,1]}^{\avg} x_{\mathsf{air}} > \mathsf{Good}) = 20.7$

% at $t = 3$, $50 - \avg(11, 30, 30) = 26.3$, 
% $\rho({\agr_{[0,1]}^{\avg} x_{\mathsf{air}} > \mathsf{Good}}) = 26.3$

% thus, $\rho(\always_{[0,3]}((\agr_{[0,1]}^{\avg} x_{\mathsf{air}} = \min \{-7, 20.7, 26.3\}) = -7$

\label{ex:quan}
\end{example}

The quantitative semantics of SaSTL inherit the two fundamental properties of STL, i.e., soundness and correctness. We give the formal definitions below. 
% The proofs are given in the appendix.  

% First, whenever $\rho(\varphi, \omega, t, l) \neq 0$, its sign indicates the satisfaction status. 

\begin{theorem}[Soundness]
Let $\varphi$ be an STL formula, $\omega$ a trace and $t$ a time,
\begin{equation*}
    \begin{array}{cc}
         \rho(\varphi, \omega, t, l) > 0 &  \Rightarrow (\omega, t, l) \sat \varphi \\
        \rho(\varphi, \omega, t, l) < 0 &  \Rightarrow (\omega, t, l) \not\sat \varphi
    \end{array}
\end{equation*}

\end{theorem}

Secondly, if $\omega$ satisfies $\varphi$  at time $t$, any other trace $\omega'$ whose point-wise distance from $\omega$ is smaller than $\rho(\varphi, \omega, t, l)$ also satisfies $\varphi$ at time $t$.  

\begin{theorem}[Correctness]
Let $\varphi$ be an STL formula, $\omega$ and $\omega'$  traces over the same time and spatial domains, and $t, l\in dom(\varphi, \omega)$, then

\begin{equation*}
    (\omega, t, l) \sat \varphi~and~||\omega - \omega'||_\infty < \rho(\varphi, \omega, t, l) \Rightarrow (\omega', t, l) \sat \varphi
\end{equation*}

\end{theorem}

In summary, the qualitative value indicates if the signal (i.e. city data) satisfies the requirement.  
The quantitative value indicates the satisfaction or dissatisfaction degree. If it is larger or equal than zero, it means that the requirement is satisfied. The larger the value is, the more the requirement is satisfied. On the contrary, if the value is smaller than zero, it means the requirement is not satisfied. The smaller the value is, the more the requirement is dissatisfied. 
}

%% file: 5_parallel.tex
\section{Efficient Monitoring for SaSTL}
\label{sect:alg}
%In this section, we provide the monitoring algorithms for the spatial operators. 
%Smart cities may have 100 or more safety and performance requirements. Each of these is stated in a SaSTL formula and monitored in real-time. 
%To enable real-time monitoring two performance improvement techniques are employed.  

In this section, we first present both Boolean and quantitative monitoring algorithms for SaSTL, then describe two optimization methods to speed up the monitoring performance.

%  and prove the time complexity of the new operators.

\subsection{Monitoring Algorithms for SaSTL}
The \textit{inputs} of the monitor are the SaSTL requirements 
$\varphi$ (including the time $t$ and location $l$), a weighted undirected graph $G$ and the temporal-spatial data $\omega$. In smart cities, the data on city states is collected continuously or periodically.
% It is a 2-dimension data, as illustrated 
% in~\figref{fig:time-spatialExample}. \figref{fig:time-spatialExample} (1) shows the noise level over a period of time and across multiple locations, x, y, and z axis indicate the location, time, and noise, respectively.  \figref{fig:time-spatialExample} (2) shows the matrix of the data, each row is a time-series data from one location and each column is a set of data from all locations at one time-stamp.

%--------------Figure: time-spatial example ----------------------------------%
% \begin{figure}[t!]
%     \centering
%     \includegraphics[width=6cm]{cityTSdata3D.png} \hspace{1cm}
%     \includegraphics[width=4.5cm]{cityTSdata2D.png}\\
%     \scriptsize{(1)  \hspace{5cm} (2)}
%     \caption{2-dimension city data over spatial and temporal domains. (Figure (1) shows the noise level over a period of time and across multiple locations, x: location, y: time (s), z: noise (dB), Figure (2) shows the matrix of the data, each row is a time-series data from one location and each column is a set of data from all locations at one time-stamp. )}
%      \vspace{-0.8cm}
%     \label{fig:time-spatialExample}
% \end{figure}

\newcontent{
For the Boolean monitoring algorithm, 
the \textit{output} for each requirement
is a Boolean value indicating whether the requirement is 
satisfied or not. For the quantitative monitoring algorithm (\algref{alg:sastlQuanti}), the \textit{output} for each requirement
is a number indicating the satisfaction degree of the requirement.  
To start with, the monitoring algorithm parses $\varphi$ to sub-formulas and calculates the satisfaction for each operation recursively. 
We derived operators $\always$ and $\eventually$ from $\until$, and operators $\everywhere$ and $\somewhere$ from $\ct^{\mathsf{op}}\sim c$, so we only show the algorithms for $\until$ and $\ct^{\mathsf{op}}\sim c$.  

% Monitoring - quantitative semantics 

\noindent
% \newcontent{
% \begin{minipage}[t]{.5\textwidth}
\begin{algorithm}[t]
\tablefontsize
\newcontent{
  \SetKwFunction{MonitorQ}{MonitorQ}
  \SetKwProg{Fn}{Function}{:}{}
%   \Fn{\MonitorQ{$\varphi,\omega, t, l, G$}}
  {
      \SetKwInOut{Input}{Input}
      \SetKwInOut{Output}{Output}
      \SetKwFor{Case}{Case}{}{}
      \Input{SaSTL Requirement $\varphi$, Signal $\omega$, Time $t$, Location $l$, weighted undirected graph $G$}
      
      \Output{Satisfaction Value $\rho$}
      
      \Begin{
            \Switch{$\varphi$} {

                   \Case{$x\sim c$}{
                        \Return $\pi_x(\omega)[t, l] - c$;
                   }

                   \Case{$\neg \varphi$}{
                        \Return - $\MonitorQ(\varphi,\omega, t, l, G)$;
                   }
 
                   \Case{$\varphi_1 \land \varphi_2$}{
                        \Return $\min(\MonitorQ(\varphi_1,\omega, t, l, G),$ \\ $\MonitorQ(\varphi_2,\omega, t, l, G))$; 
                    }
                    
                    % \Case{$ \varphi_1 U_I \varphi_2$ 
                    % \Comment*[r]{See \algref{alg:Until}.}
                    % }{ 
                    %      \Return $\mathsf{SatisfyUntilQ}(\varphi_1,\varphi_2, I, \omega, t, l, G)$;}

          \Case{$ \varphi_1 U_I \varphi_2$ }{
             $\textbf{Real}~v := -\infty$\\
             \For {$t'\in (t + I) \cap \mathbb{T}$}
             {
            $v':= \MonitorQ({\varphi_2,\omega, t',l,G})$\\
            \For {$t'' \in [t,t']$}{
            $v':= \min\{v', \MonitorQ({\varphi_2,\omega, t'',l,G})\}$
                       }
         $v = \max\{v, v'\}$         
                   
             }
        
            \Return v;

                    }

                    \Case{$\ag^{\op} x \sim c$   \Comment*[r]{See Alg. \ref{alg:AggregationQ}.} }
                    {
                          \Return $\mathsf{AggregateQ}(x, c, op, \ra, t, l, G)$;
                    }
                    
                    \Case{$\ct^{\op} \varphi \sim c$   \Comment*[r]{See Alg. \ref{alg:CountingQ}.}}
                    { 
                          \Return $\mathsf{CountingNeighboursQ}(\varphi, c, op, \ra, t, l, G) $; 
                    }
           }
       }
    }
}
\caption{\newcontent{SaSTL quantitative monitoring algorithm {$\mathsf{MonitorQ}(\varphi,\omega, t, l, G)$}}}
\label{alg:sastlQuanti}
\end{algorithm}
% \end{minipage}
% }

We present the quantitative monitoring algorithms of the operators $\ag^{\mathsf{op}}$ and $\ct^{\mathsf{op}}$  in \algref{alg:AggregationQ} and \algref{alg:CountingQ}, respectively. 
We apply distributed parallel algorithm $\mathsf{deScan()}$~\cite{ladner1980parallel} to accelerate the process of searching locations that satisfy $\ra$. 
As we can tell from the algorithms, essentially, $\ag^{\mathsf{op}}$ calculates the aggregated values on the signal over a spatial domain, while $\ct^{\mathsf{op}}$ calculates the aggregated results over spatial domain. 
For the quantitative monitoring algorithm (as presented in \algref{alg:sastlQuanti}), the \textit{output} for each requirement
is a robustness value indicating its satisfaction degree. Similar to the Boolean monitoring algorithm, the quantitative monitoring algorithm also parses $\varphi$ to sub-formulas and calculates the satisfaction for each operation recursively. }
% We present the outline of quantitative monitoring algorithm in \algref{alg:sastlQuanti}. 

% \subsection{Time Complexity}
The time complexity of monitoring the logical and temporal operators of SaSTL is the same as STL~\cite{donze2013efficient}. The time complexity to monitor classical logical operators or basic propositions such as $\neg x$, $\land$ and $x\sim c$ is $O(1)$. The time complexity to monitor  temporal operators such as $\always_{I}$, $\eventually_{I}$, $\until$ is $O(T)$, where $T$ is the total number of samples within time interval $I$.
% \begin{itemize}
%     \item The time complexity to monitor classical logical operators or basic propositions such as $\neg x$, $\land$ and $x\sim c$ is $O(1)$.
%     \item The time complexity to monitor  temporal operators such as $\always_{I}$, $\eventually_{I}$, $\until$ is $O(T)$, where $T$ is the total number of samples within time interval $I$.
% \end{itemize}
In this paper, we present the time complexity analysis for the spatial operators (Lemma \ref{lemma:spatial}) and the new SaSTL monitoring algorithm (Theorem \ref{th:timeAlg}). 
The total number of locations is denoted by $n$.  We assume that the positions of 
the locations cannot change in time (a fixed grid). We can pre-compute all the distances between locations and store them in 
an array of range trees~\cite{Lueker78} (one range tree for each 
location).
%and we can obtain $L$ in range $[d_1,d_2]$ using the table. 
We further denote the monitored formula as $\phi$, which can be represented by a syntax tree, and let $|\phi|$ denote the total number of nodes in the syntax tree (number of operators). 

\begin{lemma}[Complexity of spatial operators]
% Complexity of spatial operators: 

The time complexity to monitor at each location $l$ at time $t$ the satisfaction of a spatial operator such as  $\everywhere_{\ra}$, $\somewhere_{\ra}$, $\ag^\mathsf{op}$, and $\ct^\mathsf{op}$ is $O(log(n) + |L|)$ 
where L is the set of locations at distance within 
the range $\ra$ from $l$.
\label{lemma:spatial}
% \vspace{-0.9cm}
\end{lemma}

% \begin{proof}
% According to~\cite{Lueker78}, the time complexity to 
% retrieve a set of nodes $L$ with a distance 
% to a desired location in a range $\ra$  from a location $l$ is $O(log(n) + |L|)$.
% The aggregation and counting operations of \algref{alg:Aggregation} and \algref{alg:Counting} can be performed
% while the locations are retrieved.
% \end{proof}

% removed: Boolean monitoring algorithm

% The time complexity of the monitoring framework is given in Theorem \ref{th:timeAlg}.

\begin{theorem}
The time complexity of the SaSTL monitoring algorithm  is upper-bounded by $O(|\phi|\times T_{max}\times (log(n) + |L|_{max}))$
where $T_{max}$ is the largest number of samples of  the 
intervals considered in the temporal operators of $\phi$ 
and $|L|_{max}$ is the maximum number of locations defined 
by the spatial temporal operators of $\phi$. 
\label{th:timeAlg}
\end{theorem}

% \begin{proof}
% Following Lemma \ref{lemma:spatial}, by considering 
% $T_{max}$ the worst possible number of samples 
% that we need to consider for all possible intervals of temporal operators present
% in the formula, and $|L|_{max}$ for the worst possible number of locations that we need to consider for all possible intervals of spatial operators present
% in the formula. 
% When there are two or more operators nested, the time complexity for one operation is bounded by $O(T_{max}~(log(n) + |L|_{max}))$. 
% As there are $|\phi|$ nodes in the syntax tree of $\phi$, the time complexity of the SaSTL monitoring algorithm is bounded by the summation over all $|\phi|$ nodes, which is $O(|\phi|~T_{max}~(log(n) + |L|_{max}))$.
% %Suppose $\psi$ is an arbitrary formula, $\omega \in \mathcal{R}^{\mathbb{T} \times \mathcal{L}}$  is a trace matrix. 

% %If $\psi$ is a temporal operator, according to Lemma \ref{lemma:temporal}, the time complexity of computing $\psi$ in the syntax tree is $O(T)$.

% %If $\psi$ is a spatial operator, it at most requires to be calculated on every time-stamps in $\mathbb{T}$. According to Lemma \ref{lemma:spatial}, the time complexity of computing $\psi$ in the syntax tree is also $O(TL)$.

% %As there are exactly $|\phi|$ nodes in the syntax tree of $\phi$, the time complexity of computing algorithm Monitor($\phi$, $\omega$) is at most the summation over all $\phi$ nodes, which is $O(|\phi|TL)$.
% \end{proof}
% monitor algorithm

% removed: Boolean algorithms 

% aggregation - quantitative
\noindent
% \begin{minipage}[t]{.45\textwidth}
% \centering
\begin{algorithm}[t] 
\tablefontsize
\newcontent{
 \SetKwFunction{CS}{AggregateQ}

  \SetKwProg{Fn}{Function}{:}{}
%   \Fn{\CS{$x, c, op, \ra, \omega, t, l, G$}}{
  \Begin{
        \textbf{Real} v := 0; n := 0;
        
        \lIf{$op$ == "min"}{
             $v := \infty $
         }
        \lIf{$op$ == "max"}{
             $v := - \infty $
         }
     
      {$L^l_{\ra} := \mathsf{deScan}(l, G, \ra)$}
     
     \For {$l' \in L^l_{\ra}$}{
          
           \If{$\mathsf{op} \in \{$min, max, sum$\}$}{
                       $v$ := $\mathsf{op}(v,\pi_x(\omega)[t, l'])$;
                  }
                  \If{$\mathsf{op} == $"avg"}{
                      $v$ := $\mathsf{sum}(v,\pi_x(\omega)[t, l'])$;
                  }
                  $n := n+1$
     }
     %\If{$op$ == "avg"}{
     %   \Return $v\backslash |L^l_{[d_1,d_2]}|$;
    % }
    \lIf{$n==0$}{\Return $\infty$}
     \lIf{$\mathsf{op}$ == "avg" $\land n \neq 0$}{
                % $v :=v / n$
                % \backslash
                \Return $v / n - c$
            }
    \lIf{$\mathsf{op}$ == "sum" $\land n \neq 0$}{
               \Return $(v - c) / n$
            }
    \lElse{\Return $v - c$}
     
  }

%   }
}
\caption{\newcontent{$\mathsf{AggregateQ}(x, op, \ra, \omega, t, l, G)$}}
\label{alg:AggregationQ}
\end{algorithm}

% \end{minipage}

% \hspace{0.3cm}
% counting - quantitative
% \noindent
% \begin{minipage}[t]{.45\textwidth}
\begin{algorithm}[t]
\tablefontsize
\newcontent{
 \SetKwFunction{CS}{CountingNeighbours}
  \SetKwProg{Fn}{Function}{:}{}
   \SetKwFor{Case}{Case}{}{}
%   \Fn{\CS{$\varphi, c, op, \ra, \omega, t, l, G$}}{
      \Begin{
              \textbf{Real} $n := 0$, \textbf{List} $s := Null$;
              
               {$L^l_{\ra} := \mathsf{deScan}(l, G, \ra)$}
               
             \For {$l' \in L^l_{\ra}$}{
             s.add($\mathsf{Monitor}(\varphi,\omega, t, l', G)$)
             
                %   \If{Monitor$(\varphi,\omega, t, l, G)$ $\land$ $\mathsf{op} \in \{$min, max, sum$\}$}{
                %       $v$ := $\mathsf{op}(v,1)$;
                %   }
                %   \If{Monitor$(\varphi,\omega, t, l, G)$ $\land$ $\mathsf{op} == $"avg"}{
                %       $v$ := $\mathsf{sum}(v,1)$;
                %   }
                  $n := n+1$
             }
             
\lIf{$n==0$}{\Return $\infty$}  
\Else{
\Switch{$\mathsf{op}$}{
\Case{$\max$}{\Return $s.\max()$}
\Case{$\min$}{\Return $s.\min()$}
\Case{$\mathsf{sum}$}{\Return $s.\max(\mathsf{round}(c))$}
\Case{$\mathsf{avg}$}{\Return $s.\max(\mathsf{round}(c \times n))$}
                   }
      }
      }
%   }
 }
\caption{\newcontent{$\mathsf{CountingNeighboursQ}(x, op, \ra, \omega, t, l, G)$}}
\label{alg:CountingQ}
% \vspace{-1em}
\end{algorithm}
% \end{minipage}    

\vspace{-1cm}
\subsection{Performance Improvement of SaSTL Parsing}
%--------------Figure: time-spatial example ----------------------------------%

% The flow chart of the monitor algorithms are shown in~\figref{fig:flow}. 
To monitor a requirement, the first step is parsing the requirement to a set 
of sub formulas with their corresponding spatial-temporal ranges. 
Then, we calculate the results for the sub-formulas. The traditional parsing process of STL builds and calculates the syntax tree on the sequential order of the formula. It does not consider the complexity of each sub-formula. 
However, in many cases, especially with the PoIs specified in smart cities, checking the simpler propositional variable to quantify the spatial domain first can significantly reduce the number of temporal signals to check in a complicated formula. 
For example, the city abstracted graph in \figref{fig:eff}, the large nodes represent the locations of PoIs, among which the red ones represent the schools, and blue ones represent other PoIs. The small black nodes represent the locations of data sources (e.g. sensors). Assuming a requirement $\everywhere_{([0, +\infty), \mathsf{School})} \always_{[a,b]} (\agr^\mathsf{op}_{([0,d],\top])} \varphi \sim c)$ requires to aggregate and check $\varphi$ only nearby schools (i.e., the red circles), but it will actually check data sources of all nearby 12 nodes if one is following the traditional parsing algorithm. 
In New York City, there are about 2000 primary schools, but hundreds of thousands of PoIs in total. A very large amount of computing time would be wasted in this way.

% Parsing with the applicable range is more important in monitoring the 
% spatial temporal logic, especially when the amount of the data is large, 
% because it ensures that each step only needs to check the necessary data, 
% instead of the whole set. 

% \noindent
% \begin{figure}
% \begin{minipage}[t]{.45\textwidth}
% \centering
% \vspace{-3cm}
\begin{algorithm}[h]
\tablefontsize
  \Case{$\varphi_1 \land \varphi_2$ }{
                        \Return Monitor($\varphi_1,\omega, t, l, G$) 
                        $\land$ Monitor($\varphi_2,\omega, t, l, G$); \\
                        \If{$\mathsf{cost}(\varphi_1, l, G) \leq \mathsf{cost}(\varphi_2, l, G)$}{      
                            \If {$\neg$ Monitor($\varphi_1,\omega, t, l, G$)}{
                                 \Return Monitor($\varphi_2,\omega, t, l, G$);
                            }
                            \Return True;
                        }
                        \If {$\neg$ Monitor($\varphi_2,\omega, t, l, G$)}{
                            \Return Monitor($\varphi_1,\omega, t, l, G$);
                        }
                        \Return True;
                    }
\caption{Satisfaction of $(\varphi_1 \land \varphi_2, \omega)$}
\label{alg:and}
\end{algorithm}
% \end{minipage}

% \noindent
% \begin{minipage}[t]{.45\textwidth}
\begin{figure}
    \centering
    \includegraphics[width=5.6cm]{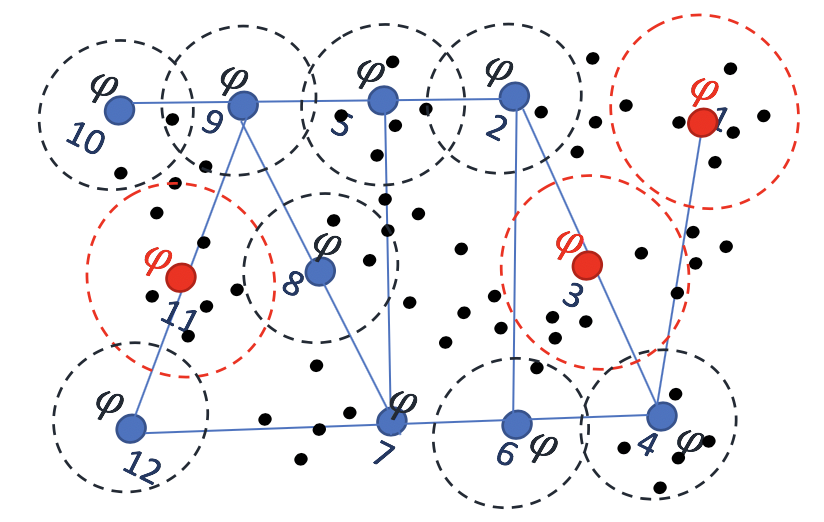} \\
    \vspace{-0.2cm}
    \captionof{figure}{An example of city abstracted graph. A requirement is $\everywhere_{([0, +\infty), \mathsf{School})} \always_{[a,b]} (\agr^\mathsf{op}_{([0,d], \top)} \varphi \sim c)$ (The large nodes represent the locations of PoIs, among which the red ones represent the schools, and blue ones represent other PoIs. The small black nodes represent the locations of data sources.)}\label{fig:eff}
    %  \vspace{-0.6em}
    % \end{subfigure}
\end{figure}
% \end{minipage}
% \end{figure}

To deal with this problem,
we now introduce a monitoring cost function $\mathsf{cost}: \Phi \times L \times G_L \rightarrow \mathbb{R}^+ $, where $\Phi$ is the set of all the possible SaSTL formulas,
$L$ is the set of locations, $G_L$ is the set of all the possible undirected graphs with $L$ locations. The cost function for $\varphi$ is defined as:

\vspace{-0.3cm}
\begin{equation*}
\footnotesize
\begin{split}
    &\mathsf{cost}(\varphi, l, G) =   \\
        &\begin{cases}
        1 & \mathsf{if}~\varphi:=p \lor \varphi:=x\sim c \lor \varphi:=\mathsf{True} \\
         1+\mathsf{cost}(\varphi_1, l, G) & \mathsf{if}~\varphi:= \neg \varphi_1\\
         \mathsf{cost}(\varphi_1, l, G) + \mathsf{cost}(\varphi_2, l, G) & \mathsf{if}~\varphi := \varphi_1 * \varphi_2, *\in\{ \land, \until\}\\
         |L^l_{\ra}| &\mathsf{if}~ \varphi := \ag^\mathsf{op} x\sim c \\
         |L^l_{\ra}| \mathsf{cost}(\varphi_1, l, G)&\mathsf{if}~ \varphi := \ct^\mathsf{op} \varphi_1 \sim c\\
        \end{cases} 
\end{split}
\end{equation*}

Using the above function, the cost of each operation is calculated before ``switch $\varphi$'' (refer to \algref{alg:sastlQuanti}). 
The cost function measures how complex it is to monitor 
a particular SaSTL formula.  This can be used when the algorithm 
evaluates the $\land$ operator and it establishes the order 
in which the sub-formulas should be evaluated. The simpler 
sub-formula is the first to be monitored, while the more 
complex one is monitored only when the other sub-formula 
is satisfied. We update $\mathsf{monitor}(\varphi_1 \land \varphi_2,\omega)$ in \algref{alg:and}.
With this cost function, the time complexity of the monitoring algorithm is reduced to $O(|\phi|\times T_{max} \times (log(n)+|L'|_{max}))$, where $|L'|$ is the maximal number of locations that an operation is executed with the improved parsing method. The improvement is significant for city requirements, where $|L'|_{max} < 100 \times |L|_{max}$.

\subsection{Parallelization}
In traditional STL monitor algorithm, the signals are checked sequentially. For example, to see if the data streams from all locations satisfy $\everywhere_{\ra}\always_{[a,b]}\varphi$ in \figref{fig:eff}, usually, it would first check the signal from location 1 with $\always_{[a,b]}\varphi$, then location 2, and so on. At last, it calculates the result from all locations with $\everywhere_{\ra}$. In this example, checking all locations sequentially is the most time-consuming part, and it could reach over 100 locations in the field. 

% In addition, during Step (1), there is no interaction between the calculation of data from two locations. 

To reduce the computing time, we parallelize the monitoring algorithm in the spatial domain. To briefly explain the idea: instead of calculating a sub-formula ($\always_{[a,b]}\varphi$) at all locations sequentially, we distribute the tasks of monitoring independent locations to different threads and check them in parallel. 
(\algref{alg:CountingPara} presents the parallel version of the spatial counting operator $\ct$.) To start with, all satisfied locations $l' \in L^l_{\ra}$ are added to a task pool (a queue). In the mapping process, each thread retrieves monitoring tasks (i.e., for $l_i, \always_{[a,b]}\varphi$) from the queue and executes them in parallel. All threads only execute one task at one time and is assigned a new one from the pool when it finishes the last one, until all tasks are executed. Each task obtains the satisfaction of $\mathsf{Monitor}(\varphi, \omega, t, l, G)$ function, and calculates the local result $v_i$ of operation $\mathsf{op}()$. The reduce step sums all the parallel results and calculates a final result  of  $\mathsf{op}()$. 

% The basic idea of parallel monitor is to distribute and execute these sub-tasks in parallel. 

\begin{algorithm}
% \vspace{-0.5cm}
\begin{multicols}{2}
\tablefontsize
 \SetKwFunction{CS}{CountingNeighbours}
 \SetKwFunction{worker}{worker}
  \SetKwProg{Fn}{Function}{:}{}
  \Fn{\CS{$\varphi,op, \ra, \omega, t, l, G$}}{
      \Begin{
 
              paratasks = Queue();
            %   \For {$l' \in L^l_{[d_1,d_2]}$}{

            \For{$l' \in L^l_{\ra}$}{paratasks.add($l$)\;}
                
                results = Queue()\;
                
                \For{i in $1..\text{NumThreads}$}{
                Thread$_i$ $\leftarrow$ worker($\varphi,\omega, t, G$)\;
                }
                Wait()\;
                \Return op(results)\;
      }
}
  \Fn{\worker($\varphi,\omega, t, G$)}{
  \Begin{
               \textbf{Real} $v := 0$;
              
              \lIf{$op$ == "min"}
              {
                    $v := \infty $;
              }
              
              \lIf{$op$ == "max"}
              {
                  $v := - \infty $;
              }
\While{Num(tasks)>0}{
$l$ = paratasks.pop()\;
moni = Monitor($\varphi, \omega, t, l, G$)\;
v = op(v, moni)\;
}
%\If{$op$ == "avg"}{
 %                \Return $v\backslash |L^{l}_{[d_1,d_2]}|$;
  %           }
  results.add(v)
  }
}
\end{multicols}            
\caption{Parallelization of Counting of $(x, op, \ra, \omega, t, l, G)$}
\label{alg:CountingPara}
\vspace{0.3cm}
\end{algorithm}

\begin{lemma}
The time complexity of the parallelized algorithm Monitor($\phi$, $\omega$) is upper bounded by $O({|\phi|}T_{max}(log(n)+\frac{|L|_{max}}{P}))$ when distributed to $P$ threads.
\label{lemma:parallel}
\end{lemma}

% \begin{proof}
% According to Lemma \ref{lemma:cost}, the time complexity of algorithm Monitor($\phi$, $\omega$) is upper-bounded by $O(|\phi|T_{max}(log(n)+|L'|_{max}))$, 
% which equals to $O(|\phi_s|T_{max}(log(n)+|L'|_{max})) + O(|\phi_t|T_{max}(log(n)+|L'|_{max}))$, with $|\phi|=|\phi_s|+|\phi_t|$.

% With the spatial operation distributed to $P$ threads, the time complexity of the spatial operation is reduced to $O(\frac{|\phi_s|T_{max}(log(n)+|L'|_{max})}{P})$. 
% Therefore, according to Amdahl's law, the parallelized time complexity of the total Monitor algorithm is upper bounded by $O((\frac{|\phi_s|}{P}+|\phi_t|)T_{max}(log(n)+|L'|_{max}))$.
% \end{proof}

In general, the parallel monitor on the spatial domain reduces the computational time significantly. It is very helpful to support runtime monitoring and decision making, especially for a large number of requirements to be monitored in a short time. 
In practice, the computing time also depends on the complexity of temporal and spatial domains as well as the amount of data to be monitored. 
A comprehensive experimental analysis of the time complexity is presented in \sectref{sect:eva}.

% In summary, these two performance improvement techniques reduce the computational complexity of executing the monitor. Example execution times and their implications are presented in the evaluation section.

%figure of the dashboard
\begin{figure*}[htbp]
    \centering
    \includegraphics[width = \textwidth]{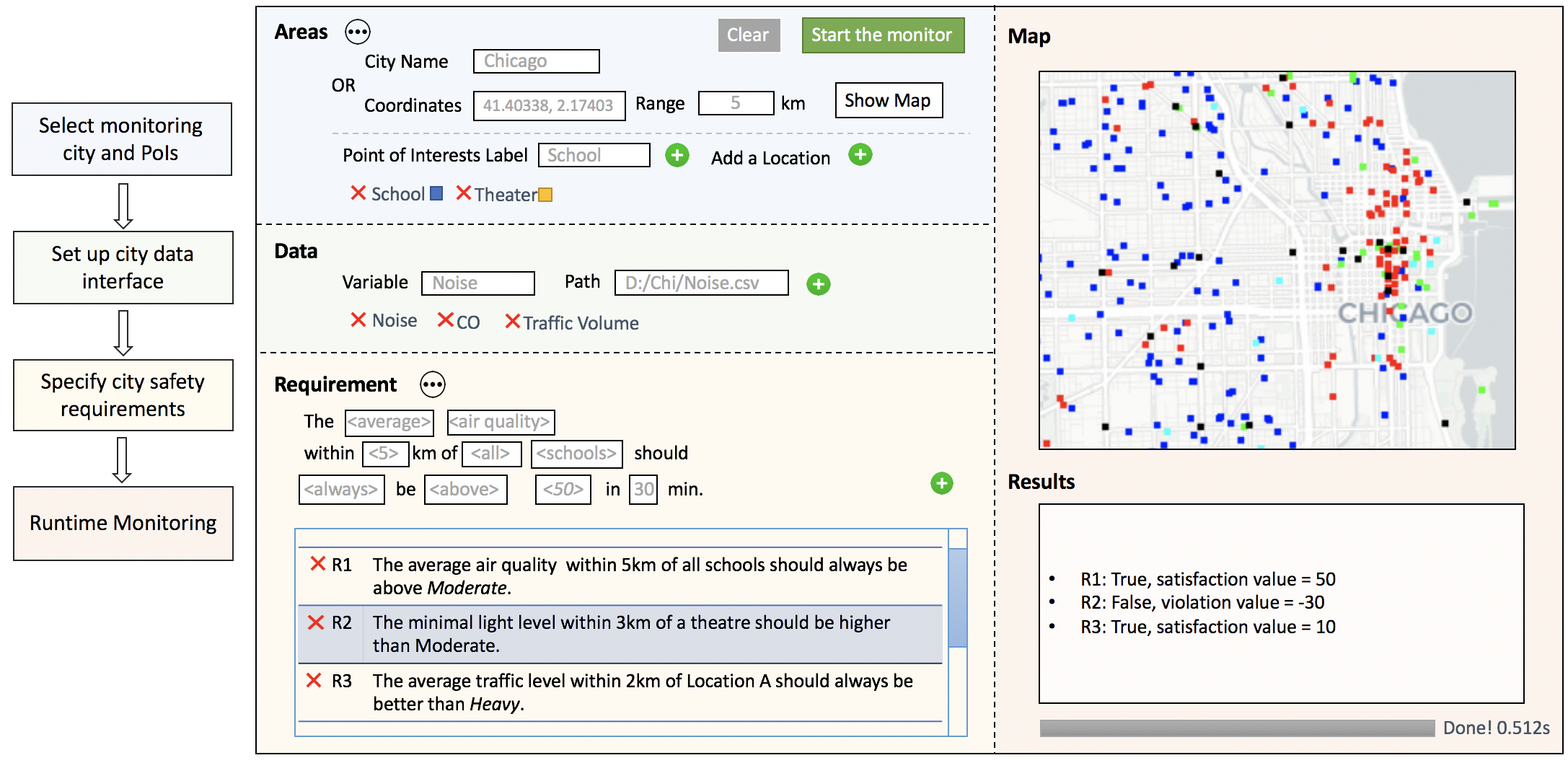}
    \caption{\newcontent{Interface of the SaSTL monitoring tool}}
    \label{fig:dash}
    % \vspace{-1em}
\end{figure*}

%% file: 6_tool.tex
\newcontent{\section{Tool for the SaSTL Monitor}
\label{sect:tool}

We develop a user-friendly prototype tool for the SaSTL monitor that can support decision making of different stakeholders in smart cities.
The interface and flowchart of the tool are shown in \figref{fig:dash}. 
The tool allows users (e.g., city decision maker, citizens) without any formal method background to check the city performance (data) with their own requirements easily in four steps.

\textit{{Step 1:} selecting the monitoring city and PoI.} 
To start with, users select the areas (such as a city, or a particular area of the city) to monitor, then choose the important labels that a requirement is involved with, such as, schools, parks, theaters, etc. 
Once selected, the important points of interest (PoIs) are shown on the map. This helps users define and verify the monitoring locations. If a location or label is not included, users are also able to add them with their GPS coordinates.  
The map displays the locations of the specified labels and sensors. Users can enlarge the map to check the distribution of sensors and PoIs and revise the requirements accordingly. 

\textit{Step 2: setting up the city data interface. }
The data of the city states collected from sensors across temporal and spatial domains are introduced to the monitor in the Data section. For the offline monitoring, users can specify the data location of each variable on the computer. For runtime monitoring, the sensing data continuously come into the computer, the data interface of which can be set up in this section. 

\textit{Step 3: specifying the city safety requirements.}
As the next important step, users specify all requirements in the requirement section. Users first select the template and then choose/fill in the essential part using the structured template language. To be noted, the entities and spatial ranges correspond to the available data variables and PoIs inputs from the areas and data sections.

We define a series of templates using structured language learning from the existing city requirements, as shown in \figref{fig:templates}. 
The goal of these templates is to help and inspire users to specify requirements precisely.  
These templates are adequate to represent
all the example requirements given in \tabref{tab:reqexamples} as well as the total set of 1,000 quantitatively-defined requirements. 
We define the templates in a recursive way. 
T is a template, and T1 and T2 are instances of T.
The elements in T are optional, i.e.  < > can be defined as blank, indicating this element is not applicable or default in this requirement. For example, an environmental requirement is written as, ``The <average> <air quality> within <1> mile of all <parks> should <always> be <above> <good>." The duration is interpreted as always (default) and there is no condition element. 
To convert a structured requirement to SaSTL, we extract the pre-defined key elements and translate them to the SaSTL formula following the rules. Meanwhile, users are also able to use the advanced features to input the city requirements in the format of the SaSTL formal formulas directly.  

\begin{figure}[htbp]
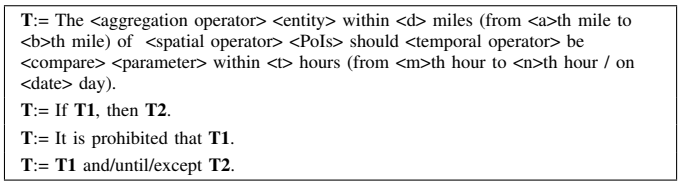

\tablefontsize
\renewcommand{\arraystretch}{1.3}
\newcontent{
	\begin{tabular}{|L{8.5cm}|}
		\hline
    \textbf{T}:= The {<aggregation operator>} <entity> within {<d> miles} ({from <a>th mile to <b>th} mile) of { <spatial operator> <PoIs>} should {<temporal operator>} be {<compare> <parameter>} {within <t> hours} ({from <m>th hour to <n>th hour} / {on <date> day}). \\ 
    \textbf{T}:= {If} \textbf{T1}, then \textbf{T2}. \\ 
    \textbf{T}:= {It is prohibited} that \textbf{T1}. \\
   \textbf{T}:= \textbf{T1} {and/until/except} \textbf{T2}. \\
    %  \textbf{T}:= \textbf{T1} {until} \textbf{T2}.\\
    % \textbf{T}:= \textbf{T1} {except} \textbf{T2}. \\
		\hline
	\end{tabular}}
	\caption{\newcontent{Templates to specify city requirements}}
    \label{fig:templates}
    \vspace{-1.5em}
\end{figure}

\textit{Step 4: runtime monitoring.} 
With all the data and requirements well defined, users can start the monitor in order to check if the incoming data from the smart city satisfies the requirements. The results are displayed with a Boolean value indicating if the requirement is satisfied and a robustness value indicating how much the requirement is satisfied or violated. 
In addition, the map also displays the monitor results visually. Two examples are shown in \figref{fig:tool_re}. The first one is monitoring an air quality requirements of high schools in Chicago, and the second one is monitoring a traffic requirement in New York City. The green circle represents the location satisfied the requirement and the red circle represents the location violates the requirement; the size of the circle represents the degree of satisfaction or violation. Users can zoom in and out the map to focus on a specific area or check the overall performance as needed (See \figref{fig:tool_re} (2)). 

%figure of the Display of the Monitoring Results
\begin{figure}[htbp]
    \centering
    \includegraphics[width = 0.4\textwidth]{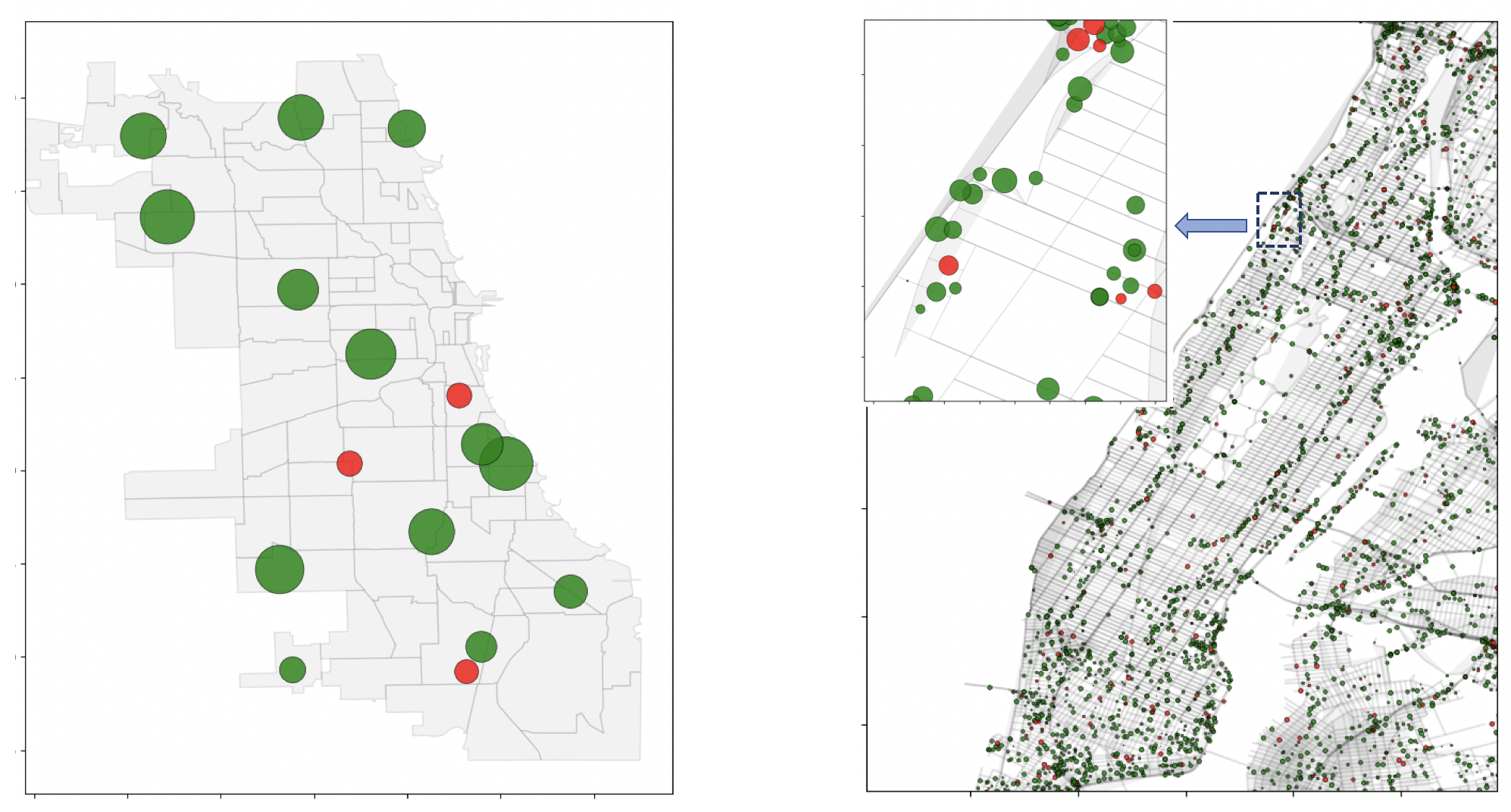}
    
    \scriptsize{
(1) Air Quality in Chicago    \hspace{1.2cm}      (2) Traffic  in New York City}
    \caption{\newcontent{Display of the Monitoring Results on the Maps (The green circle represents the location satisfied the requirement and the red circle represents the location violates the requirement; the size of the circle represents the degree of satisfaction or violation.)}}
    \label{fig:tool_re}
    % \vspace{-1em}
\end{figure}

In summary, we defined templates helping users to specify requirements to the SaSTL formal formulae. We believe these templates can not only help users to convert the requirement from English to formal formulae, they are also helpful for users to write the requirements much more specifically and precisely. The templates defined in this paper are not sufficient to cover all the city requirements, especially the new requirements coming with more and more smart services being developed. However, the approach that using structured language to specify requirements proposed in this paper is general and effective. Also, the templates are easily extended to adapt to new requirements. 

We envision this tool can be used by different stakeholders in smart cities, including but not limited to, 

\textit{City managers and decision makers}: In the city operating center, with city data collected in real time, the Tool is able to help city managers and decision makers to monitor the data at runtime. It also helps the city center to detect conflicts, and provide support for decision makers by showing the trade-offs of satisfaction degrees among potential solutions.

\textit{City planners}: City planners, either from the government to make long-term policies or from a company to make a short-term event plan, they are able to use the Tool to verify the past city data with their requirements and make plans to prevent the violations. 

\textit{Service designers}: Smart services are designed by different stakeholders including the government, companies and private parties, they are not aware of all the other services. However, with the monitor, they can test the influence of their services on the city and adjust the services to better serve the city. 

\textit{Everyday citizens}: The tool can also provide a service to the everyday citizens. Citizens without any technical background are able to specify their own requirements and check them with the city data to find out in which areas of the city and period of the day their requirements are satisfied, and make plans about their daily life. For example, a citizen can specify an environmental requirement with his/her preferred air quality index and traffic conditions, and check the city data with the requirements and make up travelling agenda accordingly.  
% In the meanwhile, we are also exploring ways to translate the requirements in English to the formal logic directly, which, however, is out of the scope of this paper.  
}

%% file: 7_evaluation.tex
\section{Evaluation}
\label{sect:eva}

% \section{Case Studies of Smart Cities}
We evaluate the SaSTL monitor by applying it to three big city application scenarios, \textit{New York}, \textit{Chicago}, and \textit{Aarhus}. 
% We provide the information of three application scenarios in \tabref{tab:three_city}.
% , including the area, data information (type, source, number of sensor nodes, time period and sampling rates), domains, monitoring variables, smart services, and evaluation metrics and baselines.  
% \figref{fig:maps} presents the partial maps of three cities, where the locations of PoIs and sensors are marked. 
\revision{The experiments are evaluated on a server machine with 20 CPUs, each core is 2.2GHz, and 4 Nvidia GeForce RTX 2080Ti GPUs. The operating system is Centos 7.}

\subsection{Runtime Monitoring of Real-Time Requirements in Chicago}
% \textit{Background}:  application scenario + data 
% \textit{Background}:  application scenario + domain + services/not + data 
% \textit{Requirements Specification}: how SaSTL helps 
% \textit{Performance}: improve city performance, algorithm performance
%--------------------------------------------------------------%

\subsubsection{Introduction}
We apply SaSTL to monitor the real-time requirements in Chicago. The framework is the same as shown in \figref{fig:overview}, where we first formalize the city requirements to SaSTL formulas and then monitor the city states with the formalized requirements. \revision{Chicago is collecting and publishing city environment data (e.g., CO, NO, O3, visible light) every day since January, 2017~\cite{arrayofthings}. In our evaluation, we emulate the Chicago data as it arrives in real time, i.e. assuming the city was operating with our SaSTL monitor. Specifically, we monitor data from 118 locations between January, 2017 and May, 2019. In addition, we incorporate the Chicago crime rate data published by the city of Chicago~\cite{chicagocrime}. The sampling rates of sensors vary by locations and variables (e.g., CO is updated every few seconds, and the crime rate map is updated by events), so we normalize the data frequency as one minute. } 
Then we specify 80 safety and performance requirements that are generated from the real requirements, and apply the SaSTL to monitor the data every 3 hours continuously to identify the requirement violations.

\noindent
% \begin{minipage}[t!]{.5\textwidth}
\begin{figure}
    \centering
    \includegraphics[width=.5\textwidth]{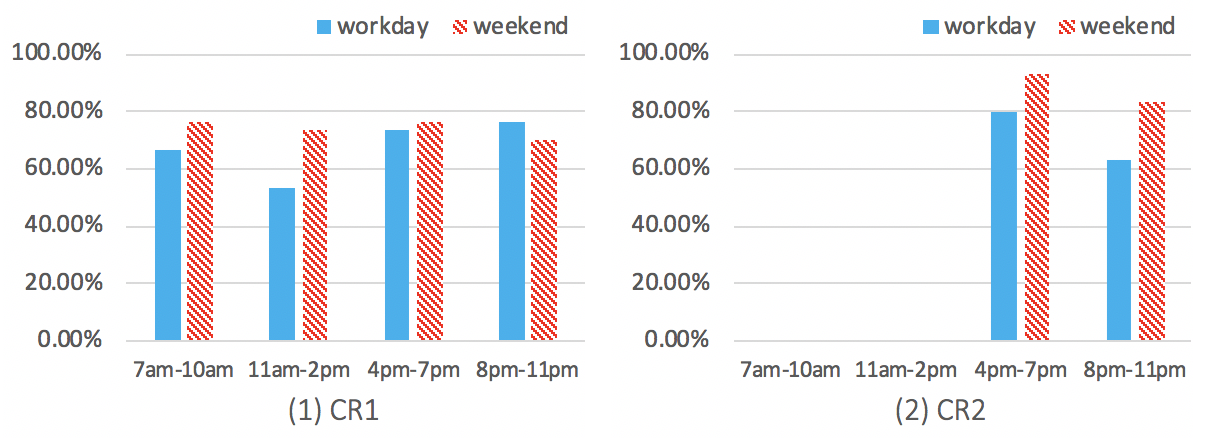}
    % ,height=0.56\textwidth
    % \vspace{-0.5cm}
    \captionof{figure}{Requirement Satisfaction Rate during Different Time Periods in Chicago}\label{fig:chicago_case}
% \end{minipage}
\end{figure}
% \hspace{0.3cm}

% \begin{minipage}[t!]{.38\textwidth}
\begin{figure}
    \centering
    \includegraphics[width=.45\textwidth]{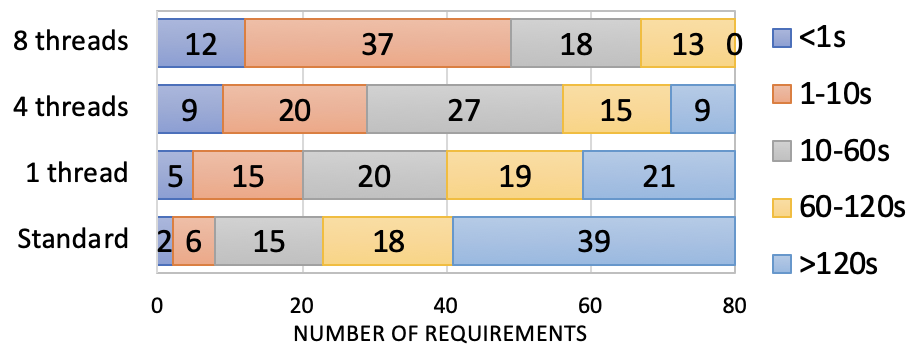}
    % ,height=0.56\textwidth
    % \vspace{-0.5cm}
    \captionof{figure}{Number of Requirements Checked on Different Computing Time}\label{fig:chicago80}
% \end{minipage}
\vspace{-0.5cm}
\end{figure}

% % "Good" AQI is 0 to 50. ...
% % "Moderate" AQI is 51 to 100. ...
% % "Unhealthy for Sensitive Groups" AQI is 101 to 150. ...
% % "Unhealthy" AQI is 151 to 200. ...
% % "Very Unhealthy" AQI is 201 to 300. ...
% % "Hazardous" AQI greater than 300.

\subsubsection{Chicago Performance}

Valuable information is identified from the monitor results of different periods during a day.
We randomly select 30 days of weekdays and 30 days of weekends. We divide the daytime of a day into 4 time periods and 3 hours per time period. 
We calculate the percentage of satisfaction (i.e., number of satisfied requirement days divides 30 days) for each time period, respectively. The results of two example requirements CR1 and CR2 are shown in \figref{fig:chicago_case}. 
CR1 specifies ``The average air quality  within 5km of all schools should always be above \textit{Moderate} in the next 3 hours.'' and is formalized as $\everywhere_{([0,+\infty),  {\mathsf{School}})}\always_{[0,3]}(\agr_{([0,5],\top)}^{\avg} x_{\mathsf{air}} > \mathsf{Moderate})$. CR2 specifies ``For the blocks with a high crime rate, the average light level within 3 km should always be \textit{High}'' and is formalized as $\everywhere_{([0,+\infty),\top)}\always_{[0,3]}  ( x_{\mathsf{Crime}} = \mathsf{High} \rightarrow \agr_{([0,3],\top)}^{\avg} x_{\mathsf{Light}} >= \mathsf{High})$. 

% \vspace{-0.2em}
% \begin{figure}[t]
%     \centering
%     \includegraphics[width=.7\textwidth]{chicago_casestudy1.png}
%     \caption{Requirement Satisfaction Rate during Different Time Periods in Chicago}
%     \label{fig:chicago_case}
% \end{figure}

% \begin{figure}[t]
%     \centering
%     \includegraphics[width=9cm]{Figure/chicago80.png}\\
%     \caption{Number of Requirements Checked on Different Computing Time (x: the number of requirements, y: the number of threads, bars: different periods of computing time)}
%     % monitoring time of requirements 
%     \label{fig:chicago80}
%     \vspace{-0.3em}
% \end{figure}

The SaSTL monitor results can be potentially used by different stakeholders. 

First, \textit{with proper requirements defined, the city decision makers are able to identify the real problems and take actions to resolve or even avoid the violations in time}. For example, from the two example requirements in \figref{fig:chicago_case}, we could see over 20\% of the time the requirements are missed everyday. 
Based on the monitoring results of requirement CR1, decision makers can take actions to redirect the traffic near schools and parks to improve the air quality. 
Another example of requirement CR2, the satisfaction is much higher (up to 33\% higher in CR2, 8pm - 11pm) over weekends than workdays.  There are more people and vehicles on the street on weekends, which as a result also increases the lighted areas. However, as shown in the figure, the city lighting in the areas with high crime rate is only 60\%. An outcome of this result for city managers is that they should pay attention to the illumination of workdays or the areas without enough light to enhance public safety.  

Second, \textit{it gives the citizens the ability to learn the city conditions and map that to their own requirements}. They can make decisions on their daily living, such as the good time to visit a park. For example, requirement CR1, 11am - 2pm has the lowest satisfaction rate of the day. 
The instantaneous air quality seems to be fine during rush hour, but it has an accumulative result that affects citizens' (especially students and elderly people) health. 
A potential suggestion for citizens who visit or exercise in the park is to avoid 11am - 2pm.

\subsubsection{Algorithm Performance}

% \vspace{-0.3em}
% \begin{figure}[t]
%     \centering
%     \includegraphics[width=9cm]{Figure/chicago80.png}\\
%     \caption{Number of Requirements Checked on Different Computing Time (x: the number of requirements, y: the number of threads, bars: different periods of computing time)}
%     % monitoring time of requirements 
%     \label{fig:chicago80}
%     \vspace{-0.3em}
% \end{figure}

\begin{table*}[t]
\caption{Safety and Performance Requirements for New York City}
% \centering
\label{table:req_ny}
\tablefontsize
\centering
\begin{tabular}{|L{0.5cm}|m{8.3cm}|L{8.cm}|}
\hline
 & \textbf{Requirement} & \textbf{SaSTL}                                                                                                           \\\hline
\textbf{NYR1} & The average noise level in the school area (within 1km) should always be less than 50dB in the next 30min.  &  
$\boxbox_{ ([0, +\infty), \mathsf{School})}\always_{[0,30]} (\agr_{([0,1], \top)}^{\avg} x_\mathsf{Noise} < 50) $ \\\hline
\textbf{NYR2} & If an accident happens, at least one of the nearby hospitals  (within 5km), its traffic condition within 2km should not reach the level of congestion in the next 60 min. &
$\everywhere_{([0,+\infty), \top)}(\mathsf{Accident}\rightarrow 
 \mathcal{C}_{([0,5],\mathsf{Hospital})}(\always_{[0,60]}(\agr_{([0,2], \top)}^{\avg} x<\mathsf{Congestion}))>0)$
\\\hline
\textbf{NYR3} & If there is an event, the max number of pedestrians waiting at an intersection should not be greater than 50 for more than 10 minutes.     & 
$\everywhere_{([0, +\infty),\top)}(\mathsf{Event} \rightarrow  \always_{[0,10]} (\agr_{([0,1],\top)}^{\max} x_\mathsf{ped} < 50))$
\\\hline
\textbf{NYR4} & At least 90\% of the streets, the PMx emission should not exceed \textit{Moderate} in 60 min.                  &  $ \mathcal{C}_{([0,+\infty),\top)}^\avg(\always_{[0,60]} (\agr_{([0,1],\top)}^{\max} x_\mathsf{PMx} < \mathsf{Moderate})) > 0.9$\\\hline
\textbf{NYR5} & If an accident happens, it should be solved within 60 min, and before that nearby (500 m) traffic should be above moderate on average and safe in worst case.    &    
$\boxbox_{([0, +\infty),\top)}(\mathsf{Accident}\rightarrow  (\agr_{([0,500],\top)}^{\avg} x_\mathsf{traffic} <\mathsf{Moderate} \land \agr_{([0,500],\top)}^{\max} x_\mathsf{traffic} < \mathsf{Safe}) \mathcal{U}_{[0,60]} \neg \mathsf{Accident})$
\\\hline                                    
\end{tabular}
\vspace{-0em}
\end{table*}

% \begin{equation*}
%     \boxbox_{([0, +\infty),\mathsf{Park})} \always_{[0,10]} (\mathcal{A}_{([0,2],\top)}^{\mathsf{avg}}(x>\mathsf{Good}))
% \end{equation*}

We count the average monitoring time taken by each requirement when monitoring for 3-hour data. Then, we divide the computing time into 5 categories, i.e., less than 1 second, 1 to 10 seconds, 10 to 60 seconds, 60 to 120 seconds, and longer than 120 seconds, and count the number of requirements under each category under the conditions of standard parsing, improved parsing with single thread, 4 threads, and 8 threads.  The results are shown in \figref{fig:chicago80}. Comparing the 1st (standard parsing) and 4th (8 threads) bar, without the improved monitoring algorithms,  for about 50\% of the requirements, each one takes more than 2 minutes to execute.  
The total time of monitoring all 80 requirements is about 2 hours, which means that the city decision maker can only take actions to resolve the violation at earliest 5 hours later. 
However, with the improved monitoring algorithms, for 49 out of 80 requirements, each one of them is executed within 60 seconds, and each one of the rest requirements is executed within 120 seconds.
The total execution time is reduced to 30 minutes, which is a reasonable time to handle as many as 80 requirements. More importantly, it illustrates the effectiveness of the parsing function and parallelization methods. Even if there are more requirements to be monitored in a real city, it is doable with our approach by increasing the number of processors.

% NYC case
\subsection{Runtime Conflict Detection and Resolution in Simulated New York City}
% \textit{Background}:  application scenario + domain + services/not + data 
% \textit{Requirements Specification}: how SaSTL helps 
% \textit{Performance}: improve city performance, algorithm performance
%--------------------------------------------------------------%

\subsubsection{Introduction}
% framework
\revision{
The framework of runtime conflict detection and resolution ~\cite{ma2016detection, ma2018cityresolver} considers a scenario where smart services send action requests to the city center, and where a simulator predicts how the requested actions change the current city states over a finite future horizon of time. Then it checks the predicted states against city requirements. If the requirements are satisfied, the requested actions will be approved to execute in the city. If there exists a requirement violation within the future horizon, a conflict is detected. CityResolver will be applied to resolve the conflicts. Details of the resolution are not the main part of this paper, please refer to CityResolver~\cite{ma2018cityresolver}. Note that with the conflicts detected and resolved, the city's future states will be affected. In this paper, we apply the SaSTL monitor to specify requirements with spatial aggregation and check the \textit{predicted spatial-temporal data} with the SaSTL formulas.

We set up a smart city simulation of New York City using the Simulation of Urban MObility (SUMO) \cite{sumo} with the traffic pattern (vehicle in-coming rate of key streets) from real city data \cite{nycopendata}, on top of which, we implement 10 services (S1: Traffic Service, S2: Emergency Service, S3: Accident Service, S4: Infrastructure Service, S5: Pedestrian Service, S6: Air Pollution Control Service, S7: PM2.5/PM10 Service, S8: Parking Service, S9: Noise Control Service, and S10: Event Service). 
% Please refer to Appendix for the detailed description of the services.  
The real-time states (including CO, NO, O3, PMx, Noise, Traffic, Pedestrian Number, Signal Lights, Emergency Vehicles, and Accident number) from the domains of environment, transportation, events and emergencies are obtained from about 10,000 simulated nodes. Then, we apply the STL Monitor as the \textit{baseline} to compare the capability of requirement specification and the ability to improve city performance. We simulate the city running for 30 days with sampling rate as 10 seconds in two control sets, one without any monitor and one with the SaSTL monitor. For the first set (no monitor), there is no requirement monitor implemented. For the second one (SaSTL monitor), five examples of different types of real-time requirements and their formalized SaSTL formulas are given in \tabref{table:req_ny}. 
}

\subsubsection{NY City Performance}

The results are shown in \tabref{tab:performance}. We measure the city performance from the domains of transportation, environment, emergency and public safety using the following metrics, the total number of violations detected (i.e.,  \revision{the total number of safety requirements violated during the whole simulation time}), the average CO (mg) emission per street, the average noise (dB) level per street, the emergency vehicles waiting time per vehicle per intersection, the average number and waiting time of vehicles waiting in an intersection per street, and the average pedestrian waiting time per intersection.
% Many times there is more than one requirement violated by a set of requested actions at one time. }

\begin{table}[t]
\caption{Comparison of the City Performance with the STL Monitor and the SaSTL Monitor}
\label{tab:performance}
\tablefontsize
\centering
\begin{tabular}{|l|r|r|}
\hline
                            & {No Monitor }   & \textbf{SaSTL Monitor} \\ \hline
\textbf{Number of Violation}         & Unknown       & \textbf{173}           \\ \hline
\textbf{Air Quality Index}                     & 67.91        & \textbf{40.18}         \\ \hline
\textbf{Noise (db) }                 & 73.32       & \textbf{41.42}         \\ \hline
\textbf{Emergency Waiting Time (s)} & 20.32          & \textbf{11.88}        \\ \hline
\textbf{Vehicle Waiting Number }    & 22.7           & \textbf{12.6}          \\ \hline
\textbf{Pedestrian Waiting Time (s)}& 190.2         & \textbf{61.1}          \\ \hline
\textbf{Vehicle Waiting Time (s)}  & 112.12         & \textbf{59.22}         \\ \hline
\end{tabular}
\vspace{-0.5em}
\end{table}

We make some observations by comparing and analyzing the monitoring results. 

First, \textit{the SaSTL monitor obtains a better city performance with fewer number of violations detected under the same scenario.}
As shown in \tabref{tab:performance}, on average, the framework of conflict detection and resolution with the SaSTL monitor improves the air quality by 40.8\%, and improves the pedestrian waiting time by 47.2\% comparing to the one without a monitor.

Second, \textit{the SaSTL monitor reveals the real city issues, helps refine the safety requirements in real time and supports improving the design of smart services. }
We also compare the number of violations on each requirement. The results (\figref{fig:pie} (1)) help the city managers to measure city's performance with smart services for different aspects, and also help policymakers to see if the requirements are too strict to be satisfied by the city and make a more realistic requirement if necessary. For example, in our 30 days simulation, apparently, NYR4 on air pollution is the one requirement that is violated by most of the smart services. 
Similarly, \figref{fig:pie} (2) shows the number of violations caused by different smart services. Most of the violations are caused by S1, S5, S6, S7, and S10. The five major services in total cause 71.3\% of the violations. City service developers can also learn from these statistics to adjust the requested actions, the inner logic and parameters of the functions of the services, so that they can design a more compatible service with more acceptable actions in the city.

 \subsubsection{Algorithm Performance}
We compare the average computing time for each requirement under four conditions, (1) using the standard parsing algorithm without the cost function, (2) improved parsing algorithm with a single thread, (3) improved parsing algorithm with spatial parallelization using 4 threads and (4) using 8 threads. The results are shown in \tabref{tab:computingtime}.

\begin{table}[t]
\tablefontsize
\caption{Computing time of requirements with standard parsing function, with improved parsing functions and different number of threads}
\label{tab:computingtime}
\begin{tabular}{|l|r| r| r| r|}
\hline
 & \textbf{Standard Parsing (s)} & \textbf{1 thread (s)} & \textbf{4 threads (s)} & \textbf{8 threads (s)} \\ \hline
\textbf{NYR1}         & 2102.13          & 140.29   & 50.31     & \textbf{26.12}     \\ \hline
\textbf{NYR2}         & 55.2             & \textbf{0.837}    & 1.023     & 0.912     \\ \hline
\textbf{NYR3}         & 69.22            & 22.25    & 7.54      & \textbf{4.822}     \\ \hline
\textbf{NYR4}         & 390.19           & 390.19   & 100.23    & \textbf{53.32}     \\ \hline
\textbf{NYR5}         & 61.76            & 61.76    & 20.25     & \textbf{15.68 }    \\ \hline
\textbf{Total}      & 2678.5           & 615.32  & 179.35   & \textbf{100.85}   \\ \hline
\end{tabular}
\vspace{-0.3em}
\end{table}

\begin{figure}[t]
    \centering
    \includegraphics[width=8cm]{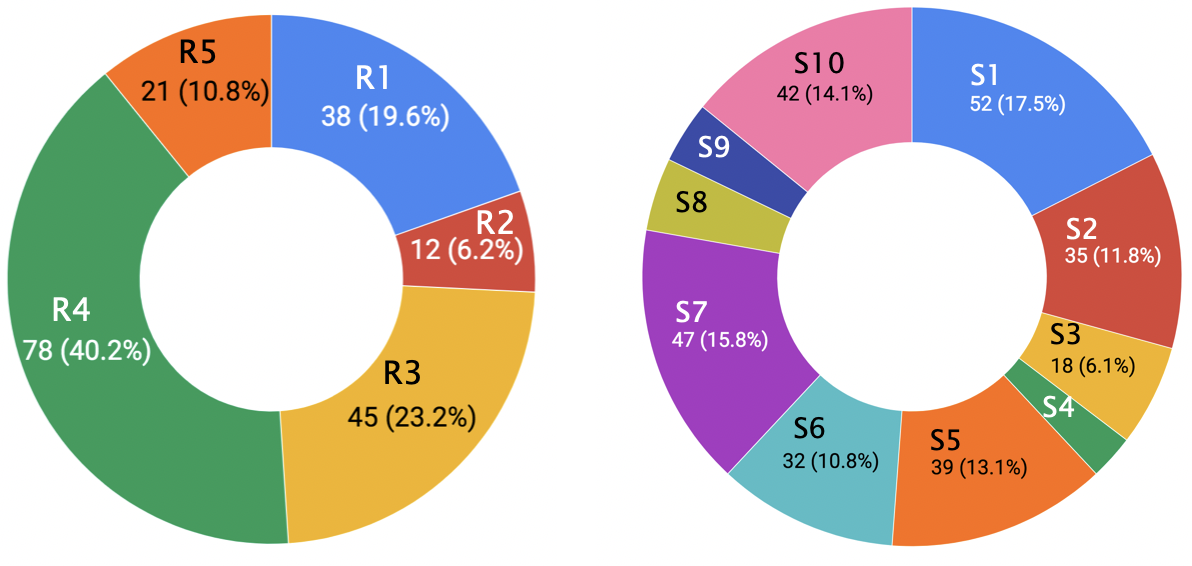}\\
    \scriptsize{(1) Requirements \hspace{1.4cm}(2) Smart Services}
    \caption{Distributions of the violations over requirements and smart services}
    \label{fig:pie}
    \vspace{-0.8em}
\end{figure}

First, the improved parsing algorithm reduces the computing time significantly for the requirement specified on PoIs, especially for NYR1 that computing time reduces from 2102.13 seconds to 140.29 seconds (about 15 times). 
Second, the parallelization over spatial operator further reduces the computing time in most of the cases. For example, for NYR1, the computing time is reduced to 26.12 seconds with 8 threads while 140.29 seconds with single thread (about 5 times). When the amount of data is very small (NYR2), the parallelization time is similar to the single thread time, but still much efficient than the standard parsing.  
% However, the efficiency of the parallelization is low when the length of the data traces to monitor are short, such as NYR2. 

% \noindent
% \begin{minipage}[t]{.48\textwidth}
\begin{figure}
    \centering
    \includegraphics[width=.45\textwidth]{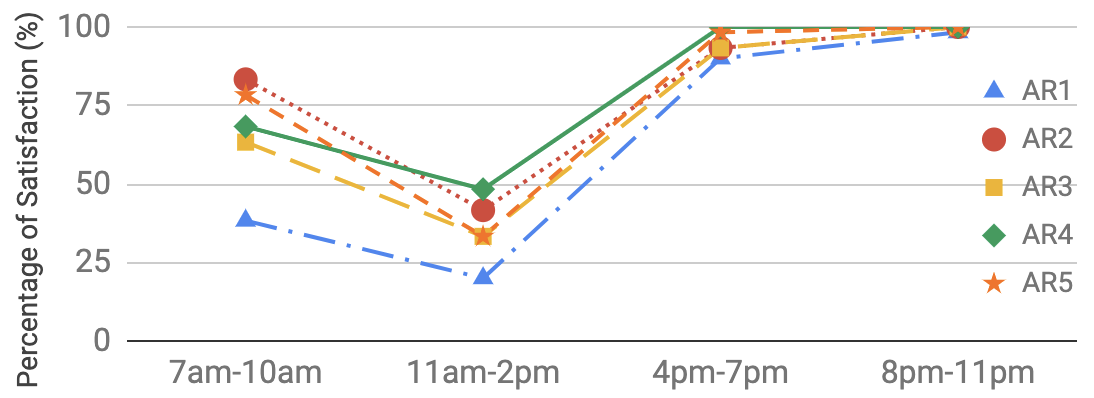}
    % \vspace{-0.2cm}
    \captionof{figure}{\newcontent{Comparisons of Satisfaction Rate on AR1 to AR5}}\label{fig:aarhus_data}
% \end{minipage}
\end{figure}
% \hspace{0.3cm}

% % \begin{minipage}[t]{.5\textwidth}
% \begin{figure}
%     \centering
%     \includegraphics[width=.5\textwidth]{ma11.png}
%     % \vspace{-0.2cm}
%     \captionof{figure}{Comparison of the Specification Coverage on 1000 Real City Requirements}\label{fig:coverage}
%     %  \vspace{-0.6em}
%     % \end{subfigure}
% \end{figure}
% % \end{minipage}

The results demonstrate the effectiveness and importance of the efficient monitoring algorithms. In the table, the total time of monitoring 5 requirements is reduced from 2678.5 seconds to 100.85 seconds. In the real world, when multiple requirements are monitored simultaneously, the improvement is extremely important for real-time monitoring.

\newcontent{
\subsection{Evaluation for Aarhus}

% \textit{Background}:  application scenario + domain + services/not + data 
% \textit{Requirements Specification}: how SaSTL helps 
% \textit{Performance}: improve city performance, algorithm performance
%--------------------------------------------------------------%

\subsubsection{Introduction}

\revision{In this case study, we monitor the past data of events and states from Aarhus to show how the SaSTL monitor helps to understand the effects caused by events and therefore aids in decision making for city events. We utilize 60 days (August to September 2014) of Aarhus city data collected simultaneously across the domains of transportation (e.g., traffic volume, parking), events (e.g., cultural events and library events) and the environment (generated pollution and weather).  
All the data were collected from 449 observation points and published by CityPulse~\cite{aarhus}. Data was collected with different sampling rates (e.g., the traffic data were aggregated by 5 minutes and events data were recorded by the event time), thus for the monitoring purpose, we normalize the data frequency as 5 minutes. }
Five safety and performance requirements and their corresponding SaSTL statements are presented with a high demand for aggregations specified for Aarhus in \tabref{table:req_aarhus}. 
Basically, AR1 to AR5 specify that when there is an event, there is a  different level of safety requirements on the traffic under different circumstances. For example, AR2 focuses on the areas nearby an event, AR3 focuses on the safety of school with an event, and R4 considers the effects from extreme weather conditions. AR5 has a big picture on all schools across the city when a large cultural event is happening. 

% The evaluation illustrates the \textit{significance of using spatial aggregation operators} and how it affects the monitor results.   
% It also shows how the SaSTL monitor helps understanding the effects caused by events and therefore aids in making better plans for events. 

% ----

\begin{table*}[htbp]
\newcontent{
\caption{\newcontent{Safety and Performance Requirements for Aarhus}}
\label{table:req_aarhus}
\tablefontsize
\begin{tabular}{|L{0.5cm}|m{8.3cm}|L{8.cm}|}
\hline
 & \textbf{Requirement} & \textbf{SaSTL}                                                                                                                            \\\hline
\textbf{AR1 }         & If there is an event, the traffic level nearby  should always be better than \textit{Moderate}.          &  $\mathsf{Event} \rightarrow \ew\always_{[0,3]} x_\text{traffic} > \mathsf{Moderate}$\\\hline
\textbf{AR2 }         & If there is an event, the average traffic level nearby should always be  better than \textit{Moderate} and the maximum traffic level nearby should be better than \textit{Safe}.                                                                   & $\mathsf{Event} \rightarrow \ew\always_{[0,3]} (\agr_{([0,1],\top)}^{\avg} x_\mathsf{traffic}> \mathsf{Moderate}  \land \agr_{([0,1],\top)}^{\max} x_\mathsf{traffic}> \mathsf{Safe})$  \\\hline
\textbf{AR3  }        & If there is an event, the average traffic near the school (3km) should always be better than \textit{Moderate} and the maximum traffic level should be better than \textit{Heavy}.                                                               & $\mathsf{Event} \rightarrow \ew\always_{[0,3]} (\agr_{([0,1],\top)}^{\avg} x_\mathsf{traffic}> \mathsf{Moderate}  \land \agr_{([0,1],\top)}^{\max} x_\mathsf{traffic}> \mathsf{Heavy}) \land \mathsf{School}$   \\\hline
\textbf{AR4  }        & If there is an event and the weather is rainy or snowy heavily, the average traffic level around school should be better than \textit{Heavy} &  $\mathsf{Event} \land \mathsf{Humidity} > 50\% \rightarrow \ew\always_{[0,3]} (\agr_{([0,1],\top)}^{\avg} x_\mathsf{traffic}> \mathsf{Heavy} ) \land \mathsf{School}$  \\\hline
\textbf{AR5  }        & With big cultural events going on the city, over the city, 80\% schools' average traffic volume nearby (3km) should always be better than \textit{Moderate}.                             & 
$\mathsf{Event} \rightarrow \ct \always_{[0,3]} (\agr_{([0,1],\top)}^{\avg} x_\mathsf{traffic}> \mathsf{Moderate} ) \land \mathsf{School} > 80\%$  \\\hline
\end{tabular}}
\vspace{-0.3cm}
\end{table*}

\subsubsection{Performance}
The monitoring results from Aarhus are shown in \figref{fig:aarhus_data}. The percentage of satisfaction equals to the number of requirement satisfied days divided by 60 days. 
The following are observations on the requirements and monitoring results.

\begin{itemize}
    % \item 
    \item Comparing the monitoring results on AR1 and AR2, AR1 has a much lower satisfaction rate. It also leads to a higher and reliable satisfaction rate. 
    \item Comparing to AR2, for the same events, AR3 moves its focus on the area nearby schools. The results, however, are lower than AR2. It means that events have more influence on the school areas, which should draw attention from the city managers. Students should reduce or avoid activities during this time when there is an event going on nearby. 
    \item During 11am to 2pm, the overall performance on all five requirements are worst, even less than 50\%. It is actually the time period right after a morning event or before an afternoon event. The monitoring results help the city managers have a better view of the distribution of effects from events.  
    \item We also find that the satisfaction rate is very high (almost 100\%) after 8pm. The reasons for that are the schools are usually closed at that time, and most of cultural and library events happen during the day. In other cities or events, the distribution will be different. However, the SaSTL monitor is general enough to help citizens and managers detect it. 
\end{itemize}

The evaluation on Aarhus shows how the SaSTL monitor helps the city to understand the effects on the city from events and make better plans for events. 
Usually, areas with an event get caught up in complicated situations, such as paralyzed traffic, long queues with a large amount of people, emergencies and accidents.  Therefore, playing back and analyzing the city data during events is extremely important for cities to avoid emergency situations for future events. }

% \subsection{Monitor Algorithm Performance Evaluation}

\section{Coverage Analysis}

We compare the specification coverage on 1000 quantitatively-specified real city requirements between STL, SSTL, STREL and SaSTL.  
The study is conducted by graduate students following the rules that if the language is able to specify the whole requirement directly with one single formula, then it is identified as True.
To be noted, another spatial STL, SpaTeL is not considered as a baseline here, because it is not applicable to most of city spatial requirements. SpaTeL is built on a quad tree, and able to specify directions rather than the distance. 

STL is only able to specify 184 out of 1000 requirements, while SSTL and STREL are able to formalize 431 requirements. 
SaSTL is able to specify 950 out of 1000 requirements.
In particular, we made the following observations from the results. 
First, 
50 requirements cannot be specified using any of the four languages because they are defined by complex math formulas that are ambiguous with missing key elements, relevant to the operations of many variables, or referring to a set of other requirements, e.g. ``follow all the requirements from Section 201.12'', etc.  
Secondly, SSTL, STREL and SaSTL outperformed STL in terms of requirements with spatial ranges, such as ``one-mile radius around the entire facility''; Third, SSTL and STREL have the same coverage on the requirements that only contain a temporal and spatial range. Comparing to SSTL and SaSTL, STREL can also be applied to dynamic graph and check requirements reachability, which is very useful in applications like wireless sensor networks, but not common in smart city requirements;
Fourth, the rest of the requirements (467 out of 1000) measure the aggregation of a set of locations, which can only be specified using SaSTL.

%% file: 8_relatedwork.tex
\section{Related Work}
\label{sect:related}

%Monitoring spatial-temporal properties over CPS executions
%was first proposed in~\cite{Talcott08} where the author 
%has introduced the notion of spatial-temporal event-based 
%model for CPS. In such model events  triggered by actions, exchange 
%of messages or a physical changes, are labeled with 
%time and space stamps and further processed by a monitor. 
% In~\cite{TVG09} 
%this concept is further elaborated, developing a spatial-temporal 
%event-based model where the space is represented
%as a 2D Cartesian coordinate system with location 
%points and location fields. 
%The approaches described in~\cite{Talcott08,TVG09} 
%provide an algorithmic framework enabling a user 
%to develop manually a monitor. However, 
%they do not provide any spatio-temporal logic 
%language enabling the specification and the 
%automatic monitoring generation.

%\todo[inline]{ben due review:last sentence of the first paragraph of section 2 is not understandable,  what do you mean by this sentence: "However, the approaches described in [19], [20] lack of a suitable a spatio-temporal specification language to express in a concise and formal way the properties of interest over the spatial-temporal event-based model."}
\newcontent{Monitoring spatial-temporal properties over CPS executions
has been initially investigated in~\cite{Talcott08,TVG09}, where the 
authors introduced a spatial-temporal event-based model for monitoring CPS. 
In this model, events are labeled with time and space stamps. These events 
can be triggered by actions, exchange of messages or physical changes.
A centralized monitor is then responsible to process all these events. Their approach 
provides an algorithmic framework enabling a user to develop manually 
a monitor, but they do not provide any spatial-temporal specification language.}
\newcontent{The literature instead offers several logic-based 
specification languages to reason about the spatial structure of concurrent systems~\cite{CC04}, medical images~\cite{BuonamiciBCLM20}, and the topological~\cite{BC02} or directional~\cite{BS10} aspects of the interacting components. 
 %In topological reasoning~\cite{BC02}, the spatial objects are sets of points and the relation between them is preserved under translation, scaling 
 %and rotation. 
 %In directional reasoning, the relation between objects depends on 
 %their relative position. 
 However, these logics are not practical for monitoring CPS, because 
 they are generally computationally complex~\cite{BS10} or even undecidable~\cite{MR99}.
 
Specification-based monitoring of spatial-temporal properties over CPS executions 
has become practical only recently with SpaTeL~\cite{bartocci2015} and SSTL~\cite{NenziBCLM15}. 
SpaTeL extends the Signal Temporal Logic~\cite{Maler2004} (STL) with the Tree Spatial Superposition Logic (TSSL)~\cite{bartocci2014,BartocciGHB18}. TSSL classifies and detects spatial patterns by reasoning over-quad trees, suitable } spatial data
structures that are constructed by recursively partitioning the space into uniform quadrants. The notion of
superposition in TSSL~\cite{BartocciGHB18} provides a way to describe statistically the distribution of discrete states in a particular
partition of the space and the spatial operators corresponding to \emph{zooming in and out in} a particular region of
the space. By nesting these operators, it is possible to specify self-similar and fractal-like structures \cite{GrosuSCWEB09} that generally characterize the patterns emerging in nature such as the electrical spiral formation in cardiac tissues~\cite{BartocciCBESG09}.
The procedure allows one to capture very complex spatial structures, but at the price of a complex formulation of spatial properties, which are in practice only learned from some template image.

\newcontent{SSTL~\cite{NenziBCLM15} extends STL with several 
spatial operators (i.e., somewhere, everywhere, and surround). 
The SSTL semantics operates on a weighted undirected graph, where
the weight on each edge represents the distance between two nodes. 
The Spatial Temporal Reach and Escape Logic (STREL)~\cite{BartocciBLN17,BartocciBLNS20} 
generalizes SSTL, by introducing two new spatial operators, (\emph{reach} and \emph{escape}), which are able to express the same spatial operators of SSTL. 
Furthermore, while SSTL can be 
applied only on static weight undirected graphs, STREL can be 
applied also to dynamic networks.
However, both SSTL and STREL do not support spatial aggregation 
operators that we show to be an important feature for monitoring smart cities.} 

%% file: 9_summary.tex
\section{Conclusion}
\label{sect:sum}

% contribution
% discussion 

In this paper, we present a novel Spatial Aggregation Signal Temporal Logic to specify and to monitor requirements of smart cities at runtime. 
We develop an efficient monitoring framework 
that optimizes the requirement parsing process and can check in parallel a SaSTL requirement over multiple data streams generated  
from thousands of sensors that are typically spatially distributed over a smart city.
SaSTL is a powerful specification language for smart cities because of its capability to monitor the city desirable features of temporal (e.g., interval), spatial (e.g., PoIs, range) and their complicated relations (e.g. always, everywhere, aggregation) between them.  
More importantly, it can coalesce many requirements into a single SaSTL formula and provide the aggregated results efficiently, which is a major advance on what smart cities do now. 
\revision{The development of 5G and 6G could better support the monitoring and communication among sensors, services and the city center.  
We believe it is a valuable step towards developing a practical smart city monitoring system 
even though there are still open issues for future work.  Furthermore, SaSTL monitor can also be easily generalized and applied to monitor other large-scale IoT deployments at runtime efficiently. In the future, we will explore its capability to specify and monitor other properties and requirements (e.g., security and privacy).}

% We evaluate the SaSTL Monitor on large-scale sensing data from two different cities. The results show the capability of SaSTL monitor algorithms to monitor large scale data streams and significantly reduce the computing time. 
% In addition, the results also
% show how our monitoring framework improve the city’s performance in simulated experiments (21.1\% on environment and 16.6\% on public safety), helping with decision making by displaying the trade-off between options, and enabling city managers and citizens to identify city requirement violations and to adapt their plans, policies and actions accordingly. 

% The future work of the SaSTL monitor includes two directions:
% First, extending SaSTL with stochastic operators to check the probability properties of city requirements, such as the chance of air quality is above Good in the next 2 hours. 
% Second, introducing the satisfaction degree particularly for city specification to compare the trade-offs between different options. The satisfaction degrees (e.g. the robustness value, time percentage of satisfaction, etc.) have been defined in STL and other works for the temporal domain. We will extend the robustness value to spatial domain in the future work. More importantly, we will explore other metrics to define the satisfaction degree for smart cities. 

%% file: appendix_iotJ.tex
\section*{Appendix}
% \section{}

\noindent\textbf{{1. Preliminaries on Signal Temporal Logic}}

% \label{sup:pre}
% \subsection{Signal Temporal Logic}
% To briefly introduce the syntax and semantics of STL, 
% we denote by $X$ and $P$ finite sets of real and propositional variables. We let
% $\omega:\Tset \rightarrow \Rset^m \times \Bset^n$ be a multi-dimensional signal, where $\Tset=[0,d) \subseteq \Rset$, $m= |X|$, $n= |P|$. 
% Given a variable $v\in X \cup P$, we denote by $\pi_v(\omega)$ the projection of $\omega$ on its component $v$.
The syntax of an STL formula $\varphi$ is usually defined as follows,

$
\varphi::=\mu \ |\ \lnot \varphi \ |\ \varphi \land \varphi \ |\ \eventually_{(a,b)} \varphi \ |\ \always_{(a,b)} \varphi \ |\ \varphi \mathbf{U}_{(a,b)} \varphi. 
$

We call $\mu$ a signal predicate, which is a formula in the form of $f(x) \ge 0$ with a signal variable $x \in \mathcal{X}$ and a function $f: \mathcal{X} \to \mathbb{R}$.
The temporal operators $\square$, $\lozenge$, and $\mathbf{U}$ denote ``always", ``eventually" and ``until", respectively. 
The bounded interval ${(a,b)}$ denotes the time interval of temporal operators. 

Below we present the formal definition of STL Boolean semantics. To informally explain the STL operations, formula $\square_{(a,b)} \varphi$ is true iff $\varphi$ is always true in the time interval ${(a,b)}$.
Formula $\lozenge_{(a,b)} \varphi$ is true iff $\varphi$ is true at sometime between $a$ and $b$.
Formula $\varphi_1 \mathbf{U}_{(a,b)} \varphi_2$ is true iff $\varphi_1$ is true until $\varphi_2$ becomes true at sometime between $a$ and $b$.

\begin{equation*}
\begin{array}{clcl}
    (\traceset, t) & \sat \mu   & \eqdef&  f(x)>0   \\    
    (\traceset, t) & \sat \neg \varphi 
        & \eqdef &  (\traceset, t) \sat \varphi \\
    (\traceset, t) & \sat \varphi_1 \land \varphi_2 
        & \eqdef&  (\traceset, t) \sat \varphi_1 \mbox{ and } (\traceset, t) \sat \varphi_2  \\
            (\traceset, t) & \sat \always_{(a,b)}
        & \eqdef& \forall t \in (a,b), (\traceset, t) \sat \varphi\\
                (\traceset, t) & \sat  \eventually_{(a,b)}  
        & \eqdef& \exists t \in (a,b) \cap \Tset, (\traceset, t) \sat \varphi \\
    (\traceset, t) & \sat \varphi_1 \until \varphi_2 
        & \eqdef& \exists t' \in (t+a, t+b) \cap \Tset, (\traceset, t') \sat \varphi_2 
          \\ &&&  \mbox{ and }  \forall t'' \in (t, t'), (\traceset, t'') \sat \varphi_1\\
\end{array}
\end{equation*}

Next, we present the formal definition of STL quantitative semantics.

\begin{alignat*}{2}
    &\rho(x \sim c, \omega, t)  
        &&= \pi_x(\omega)[t] - c \\
    &\rho(\neg \varphi, \omega, t)  
        &&= - \rho(\varphi, \omega, t) \\
    &\rho(\varphi_1 \land \varphi_2, \omega, t)  
        &&=  \min\{\rho(\varphi_1, \omega, t), \rho(\varphi_2, \omega, t) \}\\
         & \rho(\always_I \varphi, \traceset, t) && = \underset{t' \in (t, t+I)}{\min} 
   \rho(\varphi, \traceset, t')
          \\
    & \rho(\eventually_I \varphi , \traceset, t) && = \underset{t' \in (t, t+I)}{\max} 
   \rho(\varphi, \traceset, t')
          \\
    &\rho(\varphi_1 \until \varphi_2, \omega, t)  
        &&= \sup_{t'\in (t + I) \cap \mathbb{T}} (\min\{\rho(\varphi_2, \omega, t'), \\ & &&\inf_{t''\in[t,t']}(\rho(\varphi_1, \omega, t'')) \})\\
\end{alignat*}

% The quantitative semantics (i.e., the robustness values) measure the satisfaction/violation degree of the STL formula. In the evaluation section of the paper, we use it to measure the prediction performance on property satisfaction.   

\noindent
\textbf{2. Proofs}

\begingroup
\def\thetheorem{1}
\begin{theorem}[Soundness, restate]
Let $\varphi$ be an STL formula, $\omega$ a trace and $t$ a time,
\begin{equation*}
    \begin{array}{cc}
         \rho(\varphi, \omega, t, l) > 0 &  \Rightarrow (\omega, t, l) \sat \varphi \\
        \rho(\varphi, \omega, t, l) < 0 &  \Rightarrow (\omega, t, l) \not\sat \varphi
    \end{array}
\end{equation*}
\end{theorem}
\endgroup

\begin{proof}
We prove the first property $\rho(\varphi, \omega, t, l) > 0  \Rightarrow (\omega,t,l) \sat \varphi$ by induction:

First we show the soundness property hold for the predicate $\varphi:=\mu$. 
In this case, we have $\rho(\varphi, \omega, t, l) = f(x)$. Therefore, if $\rho(\varphi, \omega, t, l) >0$ we have $ f(x) >0$, that is, $(\omega,t,l) \sat \varphi$.

Case $\varphi = \neg \varphi'$: We have $\rho(\varphi, \omega, t) = -\rho(\varphi', \omega, t, l) > 0$. Therefore we have  $\rho(\varphi', \omega, t, l) < 0$, that is, $ (\omega, t, l) \not \sat \varphi'$, which is equivalent to $ (\omega, t, l) \sat \varphi$ by definition.   

Case $\varphi = \varphi_1 \land \varphi_2$: 
We have $\rho(\varphi_1 \land \varphi_2, \omega, t, l)  =  \min \{\rho(\varphi_1, \omega, t,l), \rho(\varphi_2, \omega, t,l)\} > 0 $. Therefore, we have $\rho(\varphi_1, \omega, t,l) >0$ and $\rho(\varphi_1, \omega,t,l) > 0$. Thus, $(\omega,t,l) \sat \varphi_1$ and  $(\omega,t,l) \sat \varphi_2$. By definition, we have $(\omega,t,l) \sat \varphi$.

Case $\varphi = \varphi_1 \until \varphi_2$: 
$\rho=\underset{t' \in (t, t+I)}{\max}
    \{\min\{\rho(\varphi_2, \omega, t',l), 
    \underset{t'' \in (t, t')}{\min} 
   \rho(\varphi_1, \omega, t'',l)\}\} >0$. We have $\exists t' \in (t+I), \min\{\rho(\varphi_2, \omega, t',l), 
    \underset{t'' \in (t, t')}{\min} 
   \rho(\varphi_1, \omega, t'',l)\} > 0$. Therefore, $\exists t' \in (t+I), \rho(\varphi_2, \omega, t',l) > 0 \land
    \underset{t'' \in (t, t')}{\min} 
   \rho(\varphi_1, \omega, t'',l) > 0$.  Thus, it's equivalent to   $ \exists t' \in (t+I) \cap \Tset, (\omega, t',l) \sat \varphi_2 
            \mbox{ and } \forall t'' \in (t, t'), (\omega, t'',l) \sat \varphi_1$.
            By definition, we have $(\omega, t,l) \sat \varphi$.
    
Case $\varphi = \ag^{\op} x \sim c$: we have
 $\rho(\ag^{\op} x \sim c, \omega, t, l) > 0$, which indicates $\op (\nbx) - c > 0$, following the definition, we have 
$(\ag^{\op} x \sim c, \omega, t, l)  \sat \varphi$.

Case $\varphi = \ct^{\op} \varphi \sim c$
when $\mathsf{op = max}$, we have $\max_{l' \in \nb^l}\{\rho(\varphi,\omega,t,l')\} > 0$, thus, there is at least one location $l\in \mathcal{D}$,$\rho(\varphi,\omega,t,l) > 0$, i.e., $(\omega,t,l) \sat \varphi$, therefore, $\mathsf{max}(\{g((\omega, t, l')  \sat \varphi) \ | \ l' \in \nb^l\}) > c$ ($c\in[0,1)$) is true, therefore, $(\omega, t, l) \sat \ct^{\mathsf{max}} \varphi > c$.
when $\mathsf{op = min}$, we have $\min_{l' \in \nb^l}\{\rho(\varphi,\omega,t,l')\} > 0$, thus, for any location, $\rho(\varphi,\omega,t,l) > 0$, i.e., $l\in \mathcal{D}$, $(\omega,t,l) \sat \varphi$, therefore, $\mathsf{min}(\{g((\omega, t, l')  \sat \varphi) \ | \ l' \in \nb^l\}) >c$ ($c\in[0,1)$) is true, therefore, $(\omega, t, l) \sat \ct^{\mathsf{min}} \varphi \sim c$.
When $\mathsf{op = sum}$, we have $\smallfunction(\ceil[\big]{c}, \{\rho(\varphi,\omega,t,l') \ | \ l' \in \nb^l \})>0$, thus, for at least $\ceil[\big]{c}$ locations $l$, we have $\rho(\varphi,\omega,t,l) \ | \ l \in \nb^l>0$, i.e., $\mathsf{sum}(\{g((\omega, t, l)  \sat \varphi) \ | \ l \in \nb^l\}) > c$ is true, 
therefore, $(\omega, t, l) \sat \ct^{\mathsf{sum}} \varphi > c$. Similarly, we can prove when $\mathsf{op = avg}$, if $ \smallfunction(\ceil[\big]{c \times |\nb^l |}, \{\rho(\varphi,\omega,t,l') \ | \ l' \in \nb^l \})>0$, then $(\omega, t, l) \sat \ct^{\mathsf{avg}} \varphi > c$.
% $\op(\{g((\omega, t, l')  \sat \varphi) \ | \ l' \in \nb^l\}) \sim c$
% i.e., $(\omega, t, l) \sat \ct^{\mathsf{sum}} \varphi \sim c$
        % $ \mathsf{op = avg}: \smallfunction(\ceil[\big]{c \times |\nb^l |}, \{\rho(\varphi,\omega,t,l') \ | \ l' \in \nb^l \})   $ 

\end{proof}

\begingroup
\def\thetheorem{2}
\begin{theorem}[Correctness, restate]
Let $\varphi$ be an STL formula, $\omega$ and $\omega'$  traces over the same time and spatial domains, and $t, l\in dom(\varphi, \omega)$, then

\begin{equation*}
    (\omega, t, l) \sat \varphi~and~||\omega - \omega'||_\infty < \rho(\varphi, \omega, t, l) \Rightarrow (\omega', t, l) \sat \varphi
\end{equation*}
\end{theorem}
\endgroup

\begin{proof}

First, whenever $\rho(\varphi, \omega, t, l) \neq 0$, its sign indicates the satisfaction status. 

By induction, we have the following cases:

\noindent
Case $\varphi := x \sim c$: 
We have $\rho(\varphi, \omega', t, l)= \pi_x(\omega')[t, l] - c \geq \pi_x(\omega)[t, l] - c - ||\omega - \omega'||_\infty = \rho(\varphi, \omega, t, l) - ||\omega - \omega'||_\infty> 0$. Therefore, we have $(\omega', t, l) \sat \varphi$. 

\noindent
Case $\varphi := \neg \varphi'$: We have $\rho(\varphi, \omega', t, l) = - \rho(\varphi', \omega', t, l)  $. By the inductive assumption we have $\rho(\varphi', \omega', t, l)<0$. Therefore, we have $(\omega', t, l) \sat \varphi$. 

\noindent
Case $\varphi := \varphi_1 \lor \varphi_2$: Following the condition, we have either $(\omega,t,l) \sat \varphi_1$ holds or $(\omega,t,l) \sat \varphi_2$ holds. We also have $\rho(\varphi, \omega', t,l) = \max\{\rho(\varphi_1, \omega', t, l), \rho(\varphi_2, \omega', t, l)\}$. If $(\omega,t,l) \sat \varphi_1$, by the inductive assumption we have $\rho(\varphi_1, \omega', t, l) > 0$. Therefore, $\rho(\varphi, \omega, t,l) > 0$. Similarly, if $(\omega,t,l) \sat \varphi_2$, by the inductive assumption we have $\rho(\varphi_2, \omega', t, l) > 0$. Therefore, we have $(\omega', t, l) \sat \varphi$. 

\noindent
Case $\varphi = \varphi_1 \until \varphi_2$: As $(\omega, t, l) \sat \varphi$, there exists $t'$ that  $\forall t'' \in (t, t'), \rho(\varphi_1, \omega, t'', l) \geq \rho(\varphi, \omega, t, l)$ and $\rho(\varphi_2, \omega, t', l) \geq \rho(\varphi, \omega, t, l)$. By the inductive assumption, we have $(\omega', t', l) \sat \varphi_2$ and $\forall t'' \in (t, t'), (\omega', t'', l) \sat \varphi_1$. Therefore, we have $(\omega', t, l) \sat \phi$.

\noindent   
Case $\varphi = \ag^{\op} x \sim c$: 
\begin{itemize}
    \item[-] When $\mathsf{op=sum}$, $\rho(\phi, \omega', t, l) = \frac{\mathsf{sum} (\alpha_{\ra}^x(\omega', t, l)) - c}{|\alpha_{\ra}^x(\omega', t, l)|} \geq \frac{\mathsf{sum} (\alpha_{\ra}^x(\omega, t, l)) - c - \sum_{d \in \nbx} ||\omega-\omega'||_{\infty}}{|\alpha_{\ra}^x(\omega, t, l)|} = \rho(\phi, \omega, t, l) - ||\omega-\omega'||_{\infty} > 0$. Therefore, we have $(\omega', t, l) \sat \phi$.
    \item[-] When $\mathsf{op \neq sum}$, we first show that $\mathsf{op}(\alpha_{\ra}^x(\omega, t, l)) - \mathsf{op}(\alpha_{\ra}^x(\omega', t, l)) \leq ||\omega-\omega'||_{\infty}$. Recall the definition that $\nbx:=\{\pi_x(\omega)[t, l'] \ | \ l' \in \nb^l \mbox{ and } \pi_x(\omega)[t, l'] \neq \bot\}$. For any combination of t and l, $\pi_x(\omega)[t, l] \leq \pi_x(\omega')[t, l] + ||\omega - \omega'||_{\infty}$. As all the items of $\nbx$ holds the property, for the operators max, min and avg, $\mathsf{op}(\alpha_{\ra}^x(\omega, t, l)) - \mathsf{op}(\alpha_{\ra}^x(\omega', t, l)) \leq ||\omega-\omega'||_{\infty}$.
    
    Therefore we have $\rho(\phi, \omega', t, l) = \mathsf{op}(\alpha_{\ra}^x(\omega', t, l)) - c \geq \mathsf{op}(\nbx) - ||\omega-\omega'||_{\infty} - c= \rho(\phi, \omega, t, l) - ||\omega-\omega'||_{\infty} > 0$, which indicates $(\omega', t, l) \sat \phi$.
\end{itemize}

\noindent
Case $\varphi = \ct^{\op} \varphi' \sim c$: 
\begin{itemize}
\item[-] When $\mathsf{op=sum}$, as $\rho(\ct^{\op} \varphi' \sim c, \omega, t, l) = \smallfunction(\ceil[\big]{c}, \{\rho(\varphi',\omega',t,l') \ | \ l' \in \nb^l \})$, we know that there exists at least $\ceil[\big]{c}$ different $l' \in \nb^l $ that $\rho(\varphi',\omega,t,l') \geq \rho(\ct^{\op} \varphi' \sim c, \omega, t, l) > ||\omega - \omega'||_{\infty}$. By the inductive rule, we have at least $\ceil[\big]{c}$ different $l' \in \nb^l $ that $\rho(\varphi',\omega',t,l') > 0$. Therefore, by the defintion of $\rho(\ct^{\sum} \varphi' \sim c)$ of we have $(\omega', t, l) \sat \phi$.
\item[-] Similarly when $\mathsf{op=avg}$, as $\rho(\ct^{\op} \varphi' \sim c, \omega, t, l) = \smallfunction(\ceil[\big]{c \times |\nb^l |} \{\rho(\varphi',\omega',t,l') \ | \ l' \in \nb^l \})$, we know that there exists at least $\ceil[\big]{c \times |\nb^l |}$ different $l' \in \nb^l $ that $\rho(\varphi',\omega,t,l') \geq \rho(\ct^{\op} \varphi' \sim c, \omega, t, l) > ||\omega - \omega'||_{\infty}$. By the inductive rule, we have at least $\ceil[\big]{c \times |\nb^l |}$ different $l' \in \nb^l $ that $\rho(\varphi',\omega',t,l') > 0$. Therefore, we have $(\omega', t, l) \sat \phi$.
\item[-] When $\mathsf{op=max}$, $\rho(\ct^{\op} \varphi' \sim c, \omega, t, l) = \max_{l' \in \nb^l}\{\rho(\varphi,\omega,t,l') \}$. Let $l'$ be the location that $\rho(\varphi,\omega,t,l')$ achieves maximum, we have $\rho(\varphi',\omega,t,l') \geq \rho(\ct^{\op} \varphi' \sim c, \omega, t, l) > ||\omega - \omega'||_{\infty}$. By the inductive rule, $\rho(\varphi',\omega',t,l') > 0$. Therefore, we have $(\omega', t, l) \sat \phi$.
\item[-] When $\mathsf{op=min}$, $\rho(\ct^{\op} \varphi' \sim c, \omega, t, l) = \min_{l' \in \nb^l}\{\rho(\varphi,\omega,t,l') \}$. We have for every $l' \in \nb^l$, $\rho(\varphi',\omega,t,l') \geq \rho(\ct^{\op} \varphi' \sim c, \omega, t, l) > ||\omega - \omega'||_{\infty}$. By the inductive rule, We have for every $l' \in \nb^l$ that $\rho(\varphi',\omega',t,l') > 0$. Therefore, we have $(\omega', t, l) \sat \phi$.
\end{itemize}
\end{proof}

% \vspace{-1cm}
\noindent
% \begin{minipage}[t]{.48\textwidth}
\begin{algorithm}[h]
\tablefontsize
  \SetKwFunction{Monitor}{MonitorB}
  \SetKwProg{Fn}{Function}{:}{}
%   \Fn{\Monitor{$\varphi,\omega, t, l, G$}}
  {
      \SetKwInOut{Input}{Input}
      \SetKwInOut{Output}{Output}
      \SetKwFor{Case}{Case}{}{}
      \Input{SaSTL Requirement $\varphi$, Signal $\omega$, Time $t$, Location $l$, weighted undirected graph $G$}
      
      \Output{Boolean Satisfaction Value}
      
      \Begin{
            \Switch{$\varphi$} {
                %   \Case{$p$}{
                %         \Return $\pi_{p}[t,l]$;
                %   }

                   \Case{$x\sim c$}{
                        \Return $\pi_x(\omega)[t, l] \sim c$;
                   }

                   \Case{$\neg \varphi$}{
                        \Return $\neg$ \Monitor($\varphi,\omega, t, l, G$);
                   }
 
                   \Case{$\varphi_1 \land \varphi_2$  \Comment*[r]{See Alg. \ref{alg:and} for an update}
                   }{
                        \Return \Monitor($\varphi_1,\omega, t, l, G$) 
                        $\land$ \Monitor($\varphi_2,\omega, t, l, G$) 
                        % \If{$\mathsf{cost}(\varphi_1, l, G) \leq \mathsf{cost}(\varphi_2, l, G)$}{      
                        %     \If {$\neg$ Monitor($\varphi_1,\omega, t, l, G$)}{
                        %          \Return Monitor($\varphi_2,\omega, t, l, G$);
                        %     }
                        %     \Return True;
                        % }
                        % \If {$\neg$ Monitor($\varphi_2,\omega, t, l, G$)}{
                        %     \Return Monitor($\varphi_1,\omega, t, l, G$);
                        % }
                        % \Return True;
                    }
                    
                    \Case{$ \varphi_1 U_I \varphi_2$ 
                    % \Comment*[r]{See \algref{alg:Until}.}
                    }{ 
                        %  \Return $\mathsf{SatisfyUntilB}(\varphi_1,\varphi_2, I, \omega, t, l, G)$;
                                     \textbf{Boolean} f := True;
             
             \For {$t'\in (t + I) \cap \mathbb{T}$}
             {
                  \If {Monitor$(\varphi_2, \omega, t',l,G)$}{
                      f := True;
                      
                      \For {$t'' \in [t,t']$}{
                            f := f $ \wedge $ Monitor$(\varphi_1, \omega, t'', l, G)$;
                            
                            \lIf{($\neg f$)}{\textbf{break}}
                      }
                      \lIf{($f$)}{\Return True}
                  }
             }
             \Return False;
                    }
 
                    % \Case{$ \varphi_1 S_I \varphi_2$}{
                    %      \Return SatisfySince($\varphi_1,\varphi_2, I, \omega, t, l, G$);
                    % }

                    \Case{$\ag^{\op} x \sim c$   \Comment*[r]{See Alg. \ref{alg:Aggregation}} }
                    {
                          \Return $\mathsf{AggregateB}(x, c, op, \ra, t, l, G)$;
                    }
                    
                    \Case{$\ct^{\op} \varphi \sim c$   \Comment*[r]{See Alg. \ref{alg:Counting}  and Alg. \ref{alg:CountingPara}}}
                    { 
                          \Return $\mathsf{CountingNeighboursB}(\varphi, c, op, \ra, t, l, G) $; 
                    }
           }
       }
    }
\caption{SaSTL Boolean monitoring algorithm MonitorB({$\varphi,\omega, t, l, G$})}
\label{alg:sastlBoolean}
\end{algorithm}
% \end{minipage}
% \hspace{0.3cm}

\begin{table*}[t]
	\caption{List of services running in simulated NYC}
	\centering
	\scriptsize
	\label{tab:ser}
	\begin{tabular}{|L{3.5cm}|m{13cm}|}
		\hline
Service & Description \\\hline 
\textbf{S1}: Traffic Service & It controls traffic signals in street intersections to relieve congestion and optimize or improve traffic performance. \\\hline
\textbf{S2}: Emergency Service& It requests green traffic signals in order to transport patients in critical conditions to hospitals as soon as possible.  \\\hline
\textbf{S3}: Accident Service & It blocks a street where some accident occurs and alert nearby vehicles to detour. \\\hline
\textbf{S4}: Infrastructure Service& It schedules infrastructure check-up and repair appointments. \\\hline
\textbf{S5}: Pedestrian Service  &It shortens the pedestrians' waiting time by adjusting traffic signals when pedestrians wait in the intersection.  \\\hline
\textbf{S6}:  Air Pollution Control& It adjusts the traffic by adjusting traffic signal and sending speed request to vehicles when CO emission is high.  \\\hline
\textbf{S7}: PM2.5/ PM10 Control & It adjusts the traffic when PM2.5/ PM10 emission is high by adjusting traffic signal and sending speed request to vehicles directly.  \\\hline
\textbf{S8}: Parking Service & It directs the driver to the nearest parking lot. \\\hline
\textbf{S9}: Noise Control & When noise level exceeds its threshold, it controls the number of vehicles going through related streets and redirect vehicles on the streets by adjusting traffic signals.  \\\hline
\textbf{S10}: Event Service & It ensures operation of a city event by blocking the lanes nearby the event.  \\\hline

	\end{tabular}
% 	\vspace{-1em}
\end{table*}

% \noindent
% \begin{minipage}[t]{0.9\textwidth}
% \centering
% aggregation - Boolean
% \begin{minipage}[t]{.45\textwidth}
% \centering
\begin{algorithm}[t]
\tablefontsize
 \SetKwFunction{CS}{AggregateB}
  \SetKwProg{Fn}{Function}{:}{}
  \Fn{\CS{$x, c, op, \ra, \omega, t, l, G$}}{
  \Begin{
        \textbf{Real} v := 0; n := 0;
        
        \lIf{$op$ == "min"}{
             $v := \infty $
         }
        \lIf{$op$ == "max"}{
             $v := - \infty $
         }
     
     \newcontent{$L^l_{\ra} := \mathsf{deScan}(l, G, \ra)$}

     \For {$l' \in L^l_{\ra}$}{
          
           \If{$\mathsf{op} \in \{$min, max, sum$\}$}{
                       $v$ := $\mathsf{op}(v,\pi_x(\omega)[t, l'])$;
                  }
                  \If{$\mathsf{op} == $"avg"}{
                      $v$ := $\mathsf{sum}(v,\pi_x(\omega)[t, l'])$;
                  }
                  $n := n+1$
     }
     %\If{$op$ == "avg"}{
     %   \Return $v\backslash |L^l_{[d_1,d_2]}|$;
    % }
     \lIf{$\mathsf{op}$ == "avg" $\land n \neq 0$}{
                $v :=v / n$
                % \backslash
            }
    \eIf{$n==0$}{\Return $\mathsf{True}$}{\Return $v \sim c$;}
     
  }

  }
  
\caption{$\mathsf{AggregateB}(x, op, \ra, \omega, t, l, G)$}
\label{alg:Aggregation}
\end{algorithm}
% \end{minipage}
% counting - Boolean
% \hspace{0.5cm}

% \begin{minipage}[t]{.5\textwidth}
% \centering
\begin{algorithm}[t] 
\tablefontsize
 \SetKwFunction{CS}{CountingNeighboursB}
  \SetKwProg{Fn}{Function}{:}{}
%   \Fn{\CS{$\varphi, c, op, \ra, \omega, t, l, G$}}{
      \Begin{
              \textbf{Real} $v := 0$; $n := 0$
              
              \lIf{$op$ == "min"}{
                    $v := \infty$
              }
              
              \lIf{$op$ == "max"}
              {
                  $v := - \infty$
              }
              \newcontent{$L^l_{\ra} := \mathsf{deScan}(l, G, \ra)$}
              
              \For {$l' \in L^l_{\ra}$}{
                  \If{Monitor$(\varphi,\omega, t, l, G)$ $\land$ $\mathsf{op} \in \{$min, max, sum$\}$}{
                      $v$ := $\mathsf{op}(v,1)$;
                  }
                  \If{Monitor$(\varphi,\omega, t, l, G)$ $\land$ $\mathsf{op} == $"avg"}{
                      $v$ := $\mathsf{sum}(v,1)$;
                  }
                  $n := n+1$
             }

            \lIf{$\mathsf{op}$ == "avg" $\land n \neq 0$}{
                $v :=v / n$
                % \backslash
            }
        \eIf{$n==0$}{\Return $\mathsf{True}$}{\Return $v \sim c$;}

      }

% \caption{Counting with operation $op$ the neighbours of the node l within $d_1,d_2$ that satisfy $\varphi$ at time $t$}
\caption{$\mathsf{CountingNeighboursB}(x, op, \ra, \omega, t, l, G)$}
\label{alg:Counting}
% \vspace{-1em}
\end{algorithm}
% \end{minipage}
% \end{minipage}

\begin{lemma}[Complexity of spatial operators, restate]
% Complexity of spatial operators: 

The time complexity to monitor at each location $l$ at time $t$ the satisfaction of a spatial operator such as  $\everywhere_{\ra}$, $\somewhere_{\ra}$, $\ag^\mathsf{op}$, and $\ct^\mathsf{op}$ is $O(log(n) + |L|)$ 
where L is the set of locations at distance within 
the range $\ra$ from $l$.
\label{lemma:spatial}
\end{lemma}

\begin{proof}
According to~\cite{Lueker78}, the time complexity to 
retrieve a set of nodes $L$ with a distance 
to a desired location in a range $\ra$  from a location $l$ is $O(log(n) + |L|)$.
The aggregation and counting operations of \algref{alg:Aggregation} and \algref{alg:Counting} can be performed
while the locations are retrieved.
\end{proof}

\begin{theorem}
The time complexity of the SaSTL monitoring algorithm  is upper-bounded by $O(|\phi|\times T_{max}\times (log(n) + |L|_{max}))$
where $T_{max}$ is the largest number of samples of  the 
intervals considered in the temporal operators of $\phi$ 
and $|L|_{max}$ is the maximum number of locations defined 
by the spatial temporal operators of $\phi$. 
\label{th:timeAlg}
\end{theorem}

\begin{proof}
Following Lemma \ref{lemma:spatial}, by considering 
$T_{max}$ the worst possible number of samples 
that we need to consider for all possible intervals of temporal operators present
in the formula, and $|L|_{max}$ for the worst possible number of locations that we need to consider for all possible intervals of spatial operators present
in the formula. 
When there are two or more operators nested, the time complexity for one operation is bounded by $O(T_{max}~(log(n) + |L|_{max}))$. 
As there are $|\phi|$ nodes in the syntax tree of $\phi$, the time complexity of the SaSTL monitoring algorithm is bounded by the summation over all $|\phi|$ nodes, which is $O(|\phi|~T_{max}~(log(n) + |L|_{max}))$.
\end{proof}

\noindent
\textbf{3. Monitoring Algorithms}
We presented the details of the Boolean monitoring algorithms in \algref{alg:sastlBoolean} with \algref{alg:Aggregation} for the aggregation operation and \algref{alg:Counting} for the counting operation. 

\noindent
\textbf{4. Smart Services in Simulated NYC} In the evaluation section, we set up the simulator with ten smart services. The description of these services are presented in \tabref{tab:ser}.

% maps of three cities
\begin{figure}[t]
    \centering
%      \includegraphics[width = 0.3\textwidth, trim=0 2cm 0 0, clip]{chi.png} \hspace{0.3cm}
% \includegraphics[width = 0.3\textwidth, trim=0 2cm 0 0, clip]{Figure/newy.png}\\
% \scriptsize{
% (1) Chicago    \hspace{3cm}      (2) New York }
    \includegraphics[width = 0.13\textwidth]{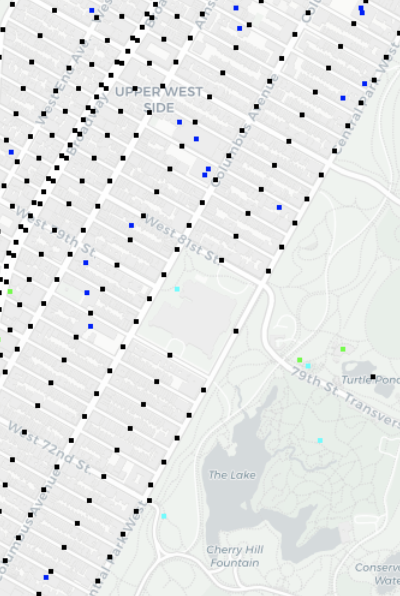} \hspace{0.1cm}
\includegraphics[width = 0.13\textwidth]{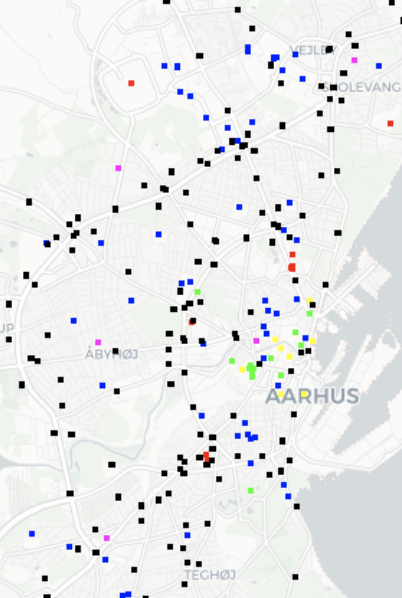} \hspace{0.1cm}
\includegraphics[width = 0.13\textwidth]{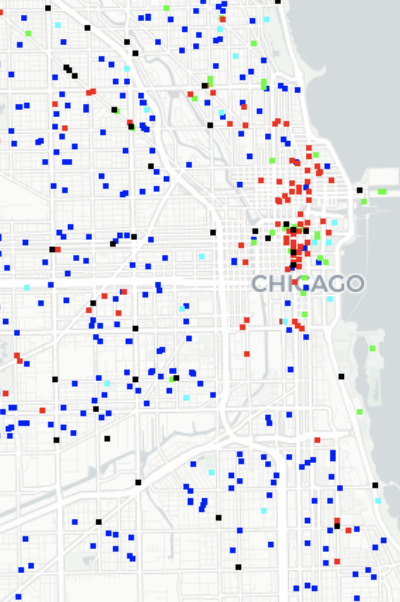}

\scriptsize{
(1) New York   \hspace{1cm}   (2) Aarhus  \hspace{1cm}       (3) Chicago }
    \caption{\newcontent{Partial Maps of Chicago, Aarhus and New York with PoIs and sensors annotated. (The black nodes represent the locations of sensors, red nodes represent the locations of hospitals, dark blue nodes represent schools, light blue nodes represent parks and green nodes represent theaters.)}}
    \label{fig:maps}
    % \vspace{-1.5em}
\end{figure}